\documentclass{llncs}

\usepackage{amssymb,amsmath}
\usepackage{xspace}
\usepackage[all]{xy} 
\SelectTips{eu}{12}

\newtheorem{rem}{Remark}
\newtheorem{prop}{Proposition}
\newtheorem{thm}{Theorem}
\newtheorem{defn}{Definition}
\newtheorem{exmp}{Example}

\newcommand{\arsto}{\rightsquigarrow}
\newcommand{\lto}{\leftarrow}
\newcommand{\dom}{\mathrm{Dom}}
\newcommand{\src}{\mathtt{src}} 
\newcommand{\tgt}{\mathtt{tgt}} 
\newcommand{\SkGr}{\mathbf{S}_{Gr}} % linear sketch for graphs 
\newcommand{\SkGrpol}{\mathbf{S}_{Gr}^\pm} % limit sketch for polarized graphs
\newcommand{\congto}{\stackrel{\simeq}{\rightarrow}}
\newcommand{\ol}{\overline}
\newcommand{\ti}{\widetilde} 
\newcommand{\rupto}[1]{\stackrel{#1}{\rightarrow}}
\newcommand{\lupto}[1]{\stackrel{#1}{\leftarrow}}
\newcommand{\longrupto}[1]{\stackrel{#1}{\longrightarrow}}
\newcommand{\spa}[5]{#1\stackrel{#2}{\leftarrow}#3\stackrel{#4}{\rightarrow}#5} % span 
\newcommand{\bN}{\mathbb{N}} % set of naturals 
\newcommand{\id}{\mathrm{id}}  % identity 
\newcommand{\nod}[1]{|#1|} % set of nodes 
\newcommand{\edg}[1]{{#1}_{\rightarrow}} % set of edges 
\newcommand{\grpol}[1]{\mathbb{#1}}   % polarized graph
\newcommand{\edgnp}[3]{{#1}_ {#2\rightarrow#3}} % set of edges from n to p etc 
\newcommand{\Set}{\mathbf{Set}} % category of sets 
\newcommand{\Gr}{\mathbf{Gr}} % category of graphs 
\newcommand{\Grpol}{\mathbf{Gr^\pm}} % category of polarized graphs 
\newcommand{\Grlab}{\mathbf{LGr}} % category of labeled graphs 
\newcommand{\Grpolab}{\mathbf{LGr^\pm}} % category of labeled polarized graphs 
\newcommand{\lset}{\mathcal{L}} % labeling set 
\newcommand{\lfun}{\mathrm{lab}}  % labeling function 
\newcommand{\parto}{\rightharpoonup} % for partial functions 
\newcommand{\catversion}{\mathcal{C}}
\newcommand{\setversion}{\mathcal{A}} 
\newcommand{\step}[3]{\xymatrix{#1 \ar@{=>}[r]^{\mathrm{#2}} & #3 \\ } }
\newcommand{\stepp}[4]{\xymatrix{#1 \ar@{=>}[r]^{\mathrm{#2}}_{#3} & #4 \\ } }
\newcommand{\catC}{\mathcal{M}}
\newcommand{\catL}{\mathcal{M}_L}
\newcommand{\catR}{\mathcal{M}_R}
\newcommand{\catP}{\mathcal{P}}
\newcommand{\funL}{\mathcal{L}} 
\newcommand{\funR}{\mathcal{R}} 
\newcommand{\funS}{\mathcal{S}}
\newcommand{\rs}{\mathrm{RS}} 
\newcommand{\norm}{\mathcal{N}}
\newcommand{\nolab}{\circ} 

\newcommand{\nodti}[1]{\nod{\ti{#1}}} 
\newcommand{\nodol}[1]{\nod{\ol{#1}}} 
\newcommand{\psqpo}{PSqPO\xspace}
\newcommand{\arpsqpo}{psqpo}
\newcommand{\depol}{\mathrm{Depol}}
\newcommand{\polar}{\mathrm{Pol}}

\title{Graph rewriting with polarized cloning} 

\author{D. Duval\inst{1} \and R. Echahed\inst{2} \and F. Prost\inst{2}}
 
\institute{ LJK\\
            B. P. 53,  F-38041 Grenoble, France \\
           \email{Dominique.Duval@imag.fr} \and
            LIG\\
            46, av F\'elix Viallet,
            F-38031 Grenoble, France\\
            \email{Rachid.Echahed@imag.fr}/
            \email{Frederic.Prost@imag.fr}}

\begin{document}

\maketitle

\begin{abstract}
  We tackle the problem of graph transformation with a particular
  focus on \emph{node cloning}. We propose a new approach to graph rewriting 
  where nodes can be cloned zero, one or more times. A node can be
  cloned together with all its incident edges, with only its outgoing
  edges, with only its incoming edges or with none of its incident
  edges. We thus subsume previous works such as the sesqui-pushout,
  the heterogeneous pushout and the adaptive star grammars approaches.
  A rewrite rule is defined as a span where the right-hand and
  left-hand sides are graphs while the interface is a polarized graph.
  A polarized graph is a graph endowed with some annotations  on nodes.
  The way a node is cloned is indicated by its polarization annotation.  We use
  these annotations for designing graph transformation with polarized
  cloning. We show how a clone of a node can be built according to the
  different possible polarizations and define a rewrite step as a
  final pullback complement followed by a pushout.  This is called the
  \emph{polarized sesqui-pushout} approach.
  We also provide an algorithmic presentation of the proposed graph
  transformation with polarized cloning.
\end{abstract}

%%%%%%%%%%%%%%%%%%%%%%%%%%%%%%%%%%%%%%%%%%%%%%%
%%%%%%%%%%%%%%%%%%%%%%%%%%%%%%%%%%%%%%%%%%%%%%%
%%%%%%%%%%%%%%%%%%%%%%%%%%%%%%%%%%%%%%%%%%%%%%%
%%%%%%%%%%%%%%%%%%%%%%%%%%%%%%%%%%%%%%%%%%%%%%%
\section{Introduction}

%%% DPO, SPO 
Graph transformation \cite{handbook1,handbook2,handbook3} extends string 
rewriting \cite{BookF93} and term rewriting \cite{BaaderN98} in several 
respects. In the literature, there are
many ways to define graphs and  graph rewriting.  
The proposed approaches  can
be gathered in two main streams: (i) the algorithmic approaches,
which define a graph rewrite step by means of the algorithms involved
in the implementation of graph transformation (see
e.g. \cite{BVG87,Echahed08b}); 
(ii) the second stream consists of the algebraic approaches, 
first proposed in the seminal paper \cite{EhrigPS73}, and which use 
categorical constructs to define graph transformation in an
abstract way. The most popular algebraic approaches are the double
pushout (DPO) and the single pushout (SPO) approaches, which can be
illustrated as follows:
$$ \begin{array}{ccc}
  
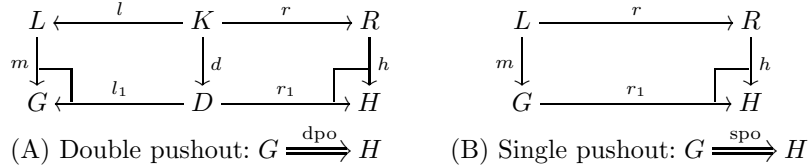
\xymatrix@C=4pc@R=1.5pc{
  L \ar[d]_{m} & K \ar[l]_{l} \ar[d]^{d}  \ar[r]^{r} 
  & R \ar[d]^{h} \\ 
  G 
  % PO
    \ar@{-}[]+R+<+6pt,+1pt>;[]+RU+<+6pt,+6pt> 
    \ar@{-}[]+U+<+1pt,+6pt>;[]+RU+<+6pt,+6pt> 
  & D \ar[l]_{l_1} \ar[r]^{r_1} 
  & H 
  % PO
    \ar@{-}[]+L+<-6pt,+1pt>;[]+LU+<-6pt,+6pt> 
    \ar@{-}[]+U+<-1pt,+6pt>;[]+LU+<-6pt,+6pt> 
  \\ 
  } &  \hspace{5mm} & 
  \xymatrix@C=6pc@R=1.5pc{
  L \ar[d]_{m} \ar[r]^{r}& R  \ar[d]^{h}  \\ 
  G  \ar[r]^{r_1} & H 
  % PO
    \ar@{-}[]+L+<-6pt,+1pt>;[]+LU+<-6pt,+6pt> 
    \ar@{-}[]+U+<-1pt,+6pt>;[]+LU+<-6pt,+6pt> 
  \\ 
  } \\
  \mbox{(A) Double pushout:} \step{G}{dpo}{H} && 
  \mbox{(B) Single pushout:} \step{G}{spo}{H} \\ 
\end{array} $$

%%% DPO
As usual, a \emph{span} is a pair of morphisms with the same source.
In the DPO approach \cite{EhrigPS73,CorradiniMREHL97}, a rule $\rho$ is
defined as a span $\rho = L\leftarrow K \rightarrow R$
and a matching is a morphism $m : L \rightarrow G$. 
A graph $G$ rewrites into a graph $H$ 
using rule $\rho$ and matching $m$ 
if and only if there exist graph morphisms $d,h,l_1,r_1$ such that 
the left and the right
squares in the diagram (A) above are pushouts. 
In general, $D$ is not unique, and conditions may be given in order to
ensure its existence, such as dangling and identification conditions 
(see e.g., \cite{CorradiniMREHL97}). 
Since graph morphisms are completely defined,  
the DPO approach is easy
to grasp, but in some cases this approach fails to specify rules with
deletion of nodes, as witnessed by the following example.
Let us
consider the reduction of the term $f(a)$ by means of the (term rewrite) rule $f(x)
\to f(b)$.  This rule can be translated into a span $f(x) \leftarrow K
\rightarrow f(b)$ for some graph $K$. When applied to $f(a)$, because
of the pushout properties, the constant $a$ must appear in $D$, hence
in $H$, although $f(b)$ is the only desired result for $H$.

%%% SPO 
In the SPO approach \cite{Rao84,Kennaway87,Lowe93,EhrigHKLRWC97}, a
rule $\rho$ is a \emph{partial} graph morphism $\rho = L \to R$
and a matching is a total morphism $m : L \rightarrow G$.
A graph $G$ rewrites into a graph $H$ using rule $\rho$ and matching $m$ 
if and only if there exist graph morphisms $h$, $r_1$ such that the 
square in the diagram (B) above is a pushout.  
This approach is appropriate to specify deletion of nodes
thanks to partial morphisms. Deletion of a node causes automatically
the deletion of all its incident edges.

%%% cloning
In this paper we are interested in graph transformation with a particular focus
on node \emph{cloning}.  
Informally, by cloning a node $n$, we mean making zero, one or more copies of
$n$ with ``some'' of its incident edges.
Roughly speaking, each copy of $n$ can be made either with all the incident edges of
$n$, with only its outgoing edges, with only its
incoming edges, or without any of its incident edges.

%%% Example (1)
Cloning a term is very common in the setting of term rewriting
systems.  Consider the rule $f(x) \rightarrow g(x,x)$. This rule
copies twice the instance of $x$ when applied over first-order
terms. In the area of graph transformation, the considered rule can be
intuitively written (following the algorithmic approach) as $f(x)
\rightarrow g(p: x, p)$ where $x$ is not cloned twice but shared
\cite{Plu98a,EcJ98a07}.  If this kind of sharing,
which is one of the main features of graph transformation, can be of
great interest in several areas such as efficient implementations of
declarative programming languages \cite{PlE93}, cloning of
nodes is another important feature which may ease graph
transformation in many real-life applications.

%%% SqPO, HPO
The classical DPO and  SPO approaches of graph transformation are
clearly not well suited to perform cloning of nodes. As far as we are
aware of, there are two algebraic attempts to deal with cloning:
the sesqui-pushout approach (SqPO) \cite{CorradiniHHK06}
and the heterogeneous pushout approach (HPO) \cite{DuvalEP09}. 
They can be illustrated as follows:
$$ \begin{array}{ccc}
  \xymatrix@C=4pc@R=1.5pc{
  L \ar[d]^{m} & K \ar[l]_{l} \ar[r]^{r} \ar[d]^{d} 
    % PB
    \ar@{-}[]+L+<-6pt,-1pt>;[]+LD+<-6pt,-6pt> 
    \ar@{-}[]+D+<-1pt,-6pt>;[]+LD+<-6pt,-6pt>  
  & R \ar[d]^{h}\\
  G & D \ar[l]_{l_1} \ar[r]^{r_1} 
  % FPBC
    \ar@{-}[]+L+<-6pt,+1pt>;[]+LU+<-6pt,+6pt> 
    \ar@{-}[]+L+<-8pt,+1pt>;[]+LU+<-8pt,+8pt> 
    \ar@{-}[]+U+<-1pt,+6pt>;[]+LU+<-6pt,+6pt> 
    \ar@{-}[]+U+<-1pt,+8pt>;[]+LU+<-8pt,+8pt> 
  & H 
  % PO
    \ar@{-}[]+L+<-6pt,+1pt>;[]+LU+<-6pt,+6pt> 
    \ar@{-}[]+U+<-1pt,+6pt>;[]+LU+<-6pt,+6pt> 
  \\
  } &  \hspace{0mm} & 
  \xymatrix@C=6pc@R=1.5pc{
  L \ar@{-->}[r]_{\tau} \ar[d]_{m} & 
      R \ar[d]^{d} \ar@{--_{>}}@/_1ex/@<-1ex>[l]_{\sigma} \\
  G \ar@{-->}[r]_{\tau_1} & 
      H \ar@{--_{>}}@/_1ex/@<-1ex>[l]_{\sigma_1} \\ 
  } \\
  \mbox{Sesqui-pushout:} \step{G}{sqpo}{H} &&
  \mbox{Heterogeneous pushout:} \step{G}{hpo}{H} \\ 
\end{array} $$

%%% SqPO 
In \cite{CorradiniHHK06}, Corradini et al.
propose the sesqui-pushout approach where a rewrite rule is a span
$L\leftarrow K \rightarrow R$ as in the DPO approach, and a rewrite step 
is obtained nearly as in the DPO approach, 
but the left square is a final pullback complement
instead of a pushout complement. 
The sesqui-pushout approach has the ability to
clone nodes with all their incident edges.

%%% HPO
In the heterogeneous pushout approach we have presented in \cite{DuvalEP09}, 
a rewrite rule is defined as a
tuple $\rho=(L,R,\tau,\sigma)$ such that the left-hand side $L$ and the
right-hand side $R$ are graphs, 
$\tau$ is a mapping from the nodes
of $L$ to the nodes of $R$ ($\tau$ does not need to be a graph morphism) and
$\sigma$ is a partial function from the nodes of $R$ to the nodes of $L$. 
A matching is a morphism of graphs $m : L \rightarrow G$. 
A graph $G$ rewrites into a graph $H$ using rule $\rho$ and matching $m$ 
if and only if the above diagram is a heterogeneous pushout 
as defined in \cite{DuvalEP09}.
Roughly speaking, this means that when $\tau(p)= n$ the incoming edges of
$p$ are redirected towards $n$ and when $\sigma(n) = p$ 
the node $n$ is instantiated as $p$ (parameter passing or cloning). 
In this approach, a node may be cloned with all its outgoing edges (and not
with all its incident edges as in the sesqui-pushout approach). 

%%% PCPO 
In this paper, we introduce the notion of polarized graphs in order to perform
node cloning. Informally, 
%which subsumes the sesqui-pushout and the heterogeneous pushout
%approaches, in providing a more flexible way for cloning.  
% RE: this remark is stated below
we define a polarized graph $\grpol{F}$ as a graph $F$ where each node
$n$ is annotated as $n^+$, $n^-$, $n^{\pm}$ or just $n$.  The rules in
our approach are made of a polarized graph $\grpol{K}$, consisting of
a graph $K$ with annotated nodes, and a span of graphs
$\spa{L}{l}{K}{r}{R}$.  Such a rule is written
$\spa{L}{l}{\grpol{K}}{r}{R}$.  Intuitively, the annotations in
$\grpol{K}$ indicate the polarized cloning actions as follows~: $n^+$
means that the node $n$ is cloned together with all its outgoing
edges, while $n^-$ means that $n$ is cloned together with all its
incoming edges. Then, $n^{\pm}$ means that $n$ is cloned together with
all its incident edges. A node $n$ without any annotation means that
$n$ is cloned without its incident edges. Thus,
polarization of nodes provides more flexible ways to perform node cloning 
 when
compared to the SqPO \cite{CorradiniHHK06} and the HPO \cite{DuvalEP09} approaches.  This
new approach is called \emph{polarized sesqui-pushout} (\psqpo for
short).  A rewrite step in the \psqpo approach has nearly the same
shape as in the SqPO approach:
$$ \begin{array}{c}
\xymatrix@C=4pc@R=1.5pc{ L \ar[d]^{m} & \grpol{K} \ar[l]_{l}
\ar[r]^{r} \ar[d]^{d} % PB 
\ar@{-}[]+L+<-6pt,-1pt>;[]+LD+<-6pt,-6pt>
\ar@{-}[]+D+<-1pt,-6pt>;[]+LD+<-6pt,-6pt> & R \ar[d]^{h}\\ G &
\grpol{D} \ar[l]_{l_1} \ar[r]^{r_1} % FPBC
\ar@{-}[]+L+<-6pt,+1pt>;[]+LU+<-6pt,+6pt>
\ar@{-}[]+L+<-8pt,+1pt>;[]+LU+<-8pt,+8pt>
\ar@{-}[]+U+<-1pt,+6pt>;[]+LU+<-6pt,+6pt>
\ar@{-}[]+U+<-1pt,+8pt>;[]+LU+<-8pt,+8pt> 
& H % PO
\ar@{-}[]+L+<-6pt,+1pt>;[]+LU+<-6pt,+6pt>
\ar@{-}[]+U+<-1pt,+6pt>;[]+LU+<-6pt,+6pt> \\ } \\ 
\mbox{Polarized sesqui-pushout:} \step{G}{\arpsqpo}{H} \\
\end{array} $$
where the left square is obtained from 
a final pullback complement 
in the category of polarized graphs
and the right square is a pushout  
in the category of graphs
(when the polarizations of $\grpol{K}$ and $\grpol{D}$ are forgotten). 

%%% outline
The paper is organized as follows. 
Section~\ref{sec:grapol} introduces the notion of polarized graphs.
Then Section~\ref{sec:rew} provides both 
an algorithmic and a categorical definition 
of graph rewriting with polarized sesqui-pushout (\psqpo). 
Our approach is adapted to labeled graphs in Section~\ref{sec:lab}.
Finally, a comparison with related
work is made in Section~\ref{sec:rel}. 
Detailed proofs can be found in the Appendix. 
We use categorical notions 
which may be found for instance in \cite{MacLane}.

%%%%%%%%%%%%%%%%%%%%%%%%%%%%%%%%%%%%%%%%%%%%%%%
%%%%%%%%%%%%%%%%%%%%%%%%%%%%%%%%%%%%%%%%%%%%%%%
\section{Graphs and polarized graphs}
\label{sec:grapol} 

In this section, we define the two kinds of graphs used in this paper:
``ordinary'' graphs in Section~\ref{subsec:gra} and polarized graphs 
in Section~\ref{subsec:pol}. 

%%%%%%%%%%%%%%%%%%%%%%%%%%%%%%%%%%%%%%%%%%%%%%%
\subsection{Graphs}
\label{subsec:gra}

\begin{defn}
\label{defi:gra-cat}
A \emph{graph} $X$ is made of a set of \emph{nodes} $\nod{X}$, a set
of \emph{edges} $\edg{X}$ and two functions \emph{source} and
\emph{target} from $\edg{X}$ to $\nod{X}$.  An edge $e$ with source
$n$ and target $p$ is denoted $n\rupto{e}p$.  The set of edges from $n$
to $p$ in $X$ is denoted $\edgnp{X}{n}{p}$.  A \emph{morphism} of
graphs $f:X\to Y$ is made of two functions (both denoted $f$)
$f:\nod{X}\to\nod{Y}$ and $f:\edg{X}\to\edg{Y}$, such that
$f(n)\rupto{f(e)}f(p)$ for each edge $n\rupto{e}p$.  This provides the
category $\Gr$ of graphs.
\end{defn}

In order to build large graphs from smaller ones, we will use the
\emph{sum} of graphs and the \emph{edge-sum} for
adding edges to a graph, as defined below using the 
symbol $+$ for the coproduct in the category of sets, 
i.e., the disjoint union of sets. 

\begin{defn}
\label{defi:gra-sum}
Given two graphs $X_1$ and
$X_2$, the \emph{sum} $X_1+X_2$ is the coproduct of $X_1$ and $X_2$
in the categry of graphs, which means that 
$\nod{X_1+X_2}=\nod{X_1}+\nod{X_2}$ and
$\edg{(X_1+X_2)}=\edg{X_1}+\edg{X_2}$ and the source and target
functions for $X_1+X_2$ are induced by the source and target functions
for $X_1$ and for $X_2$.  Given two graphs $X$ and $E$ such that
$\nod{E}\subseteq\nod{X}$, the \emph{edge-sum} $X+_eE$ is 
the pushout, in the category of graphs, of $X$ and $E$ 
over their common subgraph made of the nodes of $E$ and no edge.
This means that 
$\nod{X+_eE}=\nod{X}$ and $\edg{(X+_eE)}=\edg{X}+\edg{E}$
and the source and target functions for $X+_e E$ are induced by the
source and target functions for $X$ and for $E$.  
\end{defn}

Clearly, the precise
set of nodes of $E$ does not matter in the construction
of $X+_eE$, as long as
it contains the source and target of every edge of $E$ and is
contained in $\nod{X}$.  This notation is extended to morphisms: let
$f_1:X_1\to Y_1$ and $f_2:X_2\to Y_2$, then $f_1+f_2:X_1+X_2\to
Y_1+Y_2$ is defined piecewise from $f_1$ and $f_2$.  Similarly, let
$f:X\to Y$ and $g:E\to F$ with $\nod{E}\subseteq\nod{X}$ and
$\nod{F}\subseteq\nod{Y}$, then $f+_eg:X+_eE\to Y+_eF$ is defined as
$f$ on the nodes and piecewise from $f$ and $g$ on the edges.

\begin{rem}
\label{rem:gra-sum}
Let $X$ be a subgraph of a graph $Y$.
Let $\ol{X}$ denote the subgraph of $Y$ induced by the nodes outside $\nod{X}$ 
and $\ti{X}$ the subgraph of $Y$ induced by the edges 
which are neither in $X$ nor in $\ol{X}$, that is, the edges that are
incident to a node (at least) in $X$ but do no belong to $X$.
For all nodes $n$, $p$ in $Y$ let 
$\edgnp{\ti{X}}{n}{p}$ denote the subgraph of $Y$ induced by the 
edges from $n$ to $p$ in $\ti{X}$
(so that $\edgnp{\ti{X}}{n}{p}$ is empty whenever both $n$ and $p$ are 
in $\ol{X}$). 
Then $Y$ can be expressed as
  $ Y = (X+\ol{X}) +_e \ti{X} \quad\mbox{ with }\quad
  \edg{\ti{X}} = \sum_{n\in \nod{Y},p\in \nod{Y}} \edgnp{\ti{X}}{n}{p} $ 
which can also be written as
  $ \nod{Y} = \nod{X}+\nod{\ol{X}}$ and
  $ \edg{Y} = \edg{X}+\edg{\ol{X}} + 
               \sum_{n\in \nod{Y},p\in \nod{Y}} \edgnp{\ti{X}}{n}{p}$ 
\end{rem} 

In this paper, we use the following notion of graph \emph{matching}.
\begin{defn}
\label{defi:gra-match}
A \emph{matching of graphs} is a monomorphism of graphs. 
\end{defn}

It is easy to check that a morphism of graphs is a matching if and only if 
it is \emph{injective}, in the sense that 
both underlying functions (on nodes and on edges) are injections.  
So, up to isomorphism, every matching of graphs is an inclusion.
For simplicity of notations, we now assume that all matchings of graphs are inclusions.

%%%%%%%%%%%%%%%%%%%%%%%%%%%%%%%%%%%%%%%%%%%%%%%
\subsection{Polarized graphs}
\label{subsec:pol}

A polarized graph is a graph where every node may be polarized in the sense
that it may be marked either with a ``$+$'', with a ``$-$'', with both
``$\pm$'' or with no mark.  The polarizations will be used as cloning
instructions (in Section~\ref{subsec:alg}): roughly speaking, a node
$n^+$ is cloned with its outgoing edges, $n^-$ is cloned with its incoming
edges, $n^\pm$ is cloned  with all its incident edges and $n$ is cloned without any of its
incident edges. Moreover, every edge in a polarized graph  has
its source node (resp. target node) marked with ``$+$'' or ``$\pm$''
(resp. ``$-$'' or ``$\pm$'').

\begin{defn}
\label{defi:pol-cat}
A \emph{polarization} $X^\pm$ of a graph $X$ is a pair $X^\pm =
(\nod{X}^+,\nod{X}^-)$ of subsets of $\nod{X}$.  A node $n$ may be
denoted $n^+$ if it is in $\nod{X}^+$, $n^-$ if it is in $\nod{X}^-$
and $n^\pm$ if it is in $\nod{X}^+ \cap \nod{X}^-$.  A node $n$ of $\nod{X}$
is called \emph{neutral} if and only if 
it does not belong to $\nod{X}^+ \cup \nod{X}^-$. A \emph{polarized
  graph} $\grpol{X}=(X,X^\pm)$ is a graph $X$ together with a
polarization $X^\pm$ of $X$ such that the source of each edge $e$ of
$\edg{X}$ is in $\nod{X}^{+}$ and the target of $e$ is in
$\nod{X}^{-}$.  A \emph{morphism} of polarized graphs $f: \grpol{X}
\to \grpol{Y}$, where $\grpol{X}=(X,X^\pm)$ and $\grpol{Y}=(Y,Y^\pm)$,
is a morphism of graphs $f:X\to Y$ such that
$f(\nod{X}^+)\subseteq\nod{Y}^+$ and $f(\nod{X}^-)\subseteq\nod{Y}^-$.
This provides the category $\Grpol$ of polarized graphs.
\end{defn}

\begin{defn}
Given two polarized graphs $\grpol{X}_1$ and $\grpol{X}_2$, their
\emph{sum} is the polarized graph $\grpol{X}_1+\grpol{X}_2$ made of
the graph $X_1+X_2$ with the polarization $\nod{X_1+X_2}^+ =
\nod{X_1}^+ + \nod{X_2}^+$ and $\nod{X_1+X_2}^- = \nod{X_1}^- +
\nod{X_2}^-$. Given two polarized graphs $\grpol{X}$ and $\grpol{E}$
such that $\nod{E}\subseteq\nod{X}$, $\nod{E}^+\subseteq\nod{X}^+$ and
$\nod{E}^-\subseteq\nod{X}^-$, their \emph{edge-sum} is the polarized
graph $\grpol{X}+_e\grpol{E}$ made of the graph $X+_eE$ with the
polarization $\nod{X+_eE}^+=\nod{X}^+$ and $\nod{X+_eE}^-=\nod{X}^-$.
\end{defn}

\begin{defn}
\label{defi:pol-match}
A \emph{matching} of polarized graphs is a monomorphism $f:\grpol{X}
\to \grpol{Y}$ such that $f(\nod{X}^+)= f(\nod{X})\cap \nod{Y}^+$ and 
$f(\nod{X}^-)= f(\nod{X})\cap \nod{Y}^-$ (we say that $f$ \emph{strictly preserves} the polarization).
\end{defn}

It is clear from this definition that 
a matching of polarized graphs is 
a matching of graphs which strictly preserves the polarization.
We now assume that all matchings of polarized graphs are inclusions,
which is the case up to isomorphism. 

\begin{rem}
\label{rem:polgra-sum}
Let $f:\grpol{X}\to \grpol{Y}$ be a matching of polarized graphs.  
Analogously to Remark~\ref{rem:gra-sum},
using the fact that $f$ strictly preserves the polarization, 
we can express $\grpol{Y}$ as 
  $ \grpol{Y} = (\grpol{X}+\ol{\grpol{X}}) +_e \ti{\grpol{X}}$
  with 
  $\edg{\ti{\grpol{X}}} = \sum_{n\in \nod{\grpol{Y}},p\in \nod{\grpol{Y}}} 
          \edgnp{\ti{\grpol{X}}}{n}{p}$,   
where $\edgnp{\ti{\grpol{X}}}{n}{p}$ denotes the polarized graph 
made of the graph $\edgnp{\ti{X}}{n}{p}$ as in Remark~\ref{rem:gra-sum}
with its nodes polarized as in $\grpol{Y}$.
\end{rem}

\begin{exmp}
 \label{exam:pol}
Here is a morphism of polarized graphs which is an inclusion
although it is not a matching (the condition $f(\nod{X}^+)= f(\nod{X})\cap \nod{Y}^+$ is not fulfilled):
$$ \begin{array}{|c|c|c|}
 \cline{1-1} \cline{3-3}  
 \rule[0pt]{0pt}{20pt} 
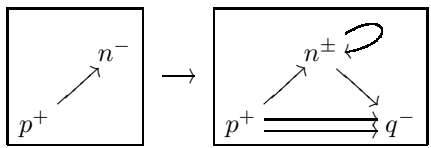
\xymatrix@=1pc{ & n^- \\
                  p^+ \ar[ru] &  \\ }
 & \xymatrix@R=.2pc@C=1pc{ \\ 
                  \ar[r] &} & 
 \xymatrix@=1pc{ & n^\pm \ar[rd] \ar@(ur,r) & \\
                  p^+ \ar[ru] \ar@<.5ex>[rr] \ar@<-.5ex>[rr] && q^-  \\ } \\ 
 \cline{1-1} \cline{3-3}  
  \end{array}   $$
\end{exmp}

\begin{defn}
\label{defi:alg-minmax}
  The \emph{underlying} 
  graph of a polarized graph $\grpol{X}=(X,X^{\pm})$ is $X$. 
  This defines a functor $\depol:\Grpol\to\Gr$.  
  The polarized graph $\grpol{X}$ \emph{induced by} a graph $X$ 
  is $\grpol{X}=(X,X^{\pm})$ where $\nod{X}^{+}=\nod{X}^{-}=\nod{X}$. 
  This defines a functor $\polar:\Gr\to\Grpol$, 
  which is a right adjoint to $\depol$. 
\end{defn}

%%%%%%%%%%%%%%%%%%%%%%%%%%%%%%%%%%%%%%%%%%%%%%%
%%%%%%%%%%%%%%%%%%%%%%%%%%%%%%%%%%%%%%%%%%%%%%%
\section{Polarized sesqui-pushout graph rewriting}
\label{sec:rew}

In this section we define the polarized sesqui-pushout approach
(\psqpo).  The idea is that the polarization of a node indicates how the
cloning acts on the edges incident to the considered node.  In order
to provide an operational intuition of our \psqpo approach, we start by
giving an algorithmic counterpart of \psqpo in Section~\ref{subsec:set},
before defining the algebraic approach in
Section~\ref{subsec:alg}. The equivalence of the two approaches is
discussed in Section~\ref{subsec:equiv}. We end this section by
stating the functoriality property (also known as vertical
composition) of the \psqpo approach in Section~\ref{subsec:functo}.

%%%%%%%%%%%%%%%%%%%%%%%%%%%%%%%%%%%%%%%%%%%%%%%
\subsection{An algorithmic definition}
\label{subsec:set}

The algorithmic version of polarized sesqui-pushout graph rewriting  
is denoted AlgoPC. 

\begin{defn}
\label{defi:set-setpc}
An \emph{AlgoPC rewrite rule} consists of a tuple $\mu=(L,R,C^+,C^-)$, 
where $L$ and $R$ are graphs and 
$C^+,C^-:\nod{L}\times\nod{R} \to \bN$ are mappings. 
Then $L$ and $R$ are called respectively the \emph{left-hand side}
and the \emph{right-hand side} 
and $C^+$, $C^-$ are called the \emph{cloning multiplicities} of $\mu$.
Let $\mu=(L,R,C^+,C^-)$ be an AlgoPC rewrite rule, 
$G$ a graph and $m:L \to G$ a matching. 
Thus $\nod{G} = \nod{L} + \nod{\ol{L}}$
and $\edg{G} = \edg{L} + \edg{\ol{L}} + \edg{\ti{L}}$.
The \emph{AlgoPC rewrite step} applying the rule $\mu$
to the matching $m$ builds the graph $H$ and the matching $h:R\to H$
such that $h$ is the inclusion and
$ \nod{H} = \nod{R} + \nod{\ol{L}}$  and 
$\edg{H} = \edg{R} + \edg{\ol{L}} + \sum_{n\in\nod{H},p\in\nod{H}}E_{n,p} $
where: 
  \begin{enumerate} 
  \item if $n\in\nod{R}$ and $p\in\nod{R}$ 
  then there is an edge $n\rupto{(e,i)} p$ in $E_{n,p}$ 
  for each edge $n_L \rupto{e} p_L$ in $\edg{\ti{L}}$
  and each $i\in\{1,\dots,C^+(n_L,n) \times C^-(p_L,p)\}$;
  \item if $n\in\nod{R}$ and $p\in\nod{\ol{L}}$
  then there is an edge $n\rupto{(e,i)} p$ in $E_{n,p}$ 
  for each edge $n_L \rupto{e} p$ in $\edg{\ti{L}}$
  and each $i\in\{1,\dots,C^+(n_L,n)\}$;
  \item if $n\in\nod{\ol{L}}$ and $p\in\nod{R}$
  then there is an edge $n\rupto{(e,i)} p$ in $E_{n,p}$ 
  for each edge $n \rupto{e} p_L$ in $\edg{\ti{L}}$
  and each $i\in\{1,\dots,C^-(p_L,p)\}$;
  \item if $n\in\nod{\ol{L}}$ and $p\in\nod{\ol{L}} $
  then $E_{n,p}$ is empty.
  \end{enumerate}
An AlgoPC rewrite step is written $\step{G}{algopc}{H}$ or 
more precisely $\stepp{m}{algopc}{\mu}{h}$. 
\end{defn}

So, when an AlgoPC rule $\mu=(L,R,C^+,C^-)$ is applied to a matching of $L$ in $G$, 
the image of $L$ in $G$ is erased and replaced by $R$,
the subgraph $\ol{L}$ remains unchanged, 
and the edges in $\ti{L}$ are handled according to the cloning multiplicities. 
The subtleties in building clones  
lie in the treatment of the edges in $\ti{L}$.

\begin{exmp}
 \label{exam:set-setpc}
 Let us consider the following rule $\mu = (L,R,C^+,C^-)$ where 
 $$ \begin{array}{|c|c|c|}
 \cline{1-1} \cline{3-3}  
 L && R \\
 \cline{1-1} \cline{3-3}  
 
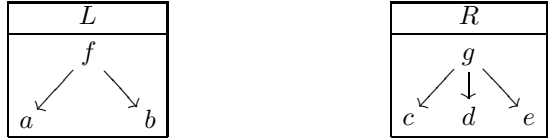
\xymatrix@=1pc{ & f \ar[dl] \ar[dr]&  \\
                  a  &  & b}
 & \qquad\qquad\qquad\qquad &
 \xymatrix@=1pc{ & g \ar[dl] \ar[dr] \ar[d] &  \\
                 c  & d & e} \\
 \cline{1-1} \cline{3-3}  
  \end{array}   $$
$C^+(a,c)=2$, $C^+(a,e)=1$, $C^-(f,g)=2$,
and every other cloning multiplicity is~0. 
Now let us consider the graphs $G$ and $H$: 
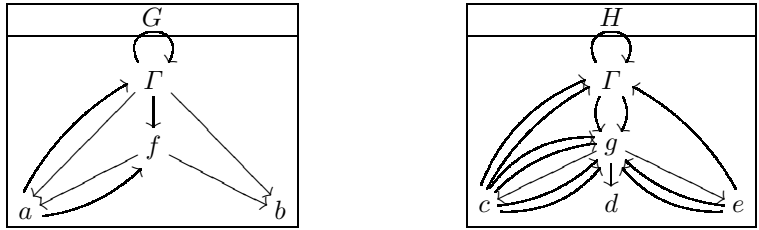
 $$ \begin{array}{|c|c|c|}
 \cline{1-1} \cline{3-3}  
 G & & H \\
 \cline{1-1} \cline{3-3}  
 \xymatrix@R=.2pc@C=3pc{\\  
      & \Gamma \ar@(ul,ur) \ar[dd] \ar@<.5ex>[ddddr] \ar@<-.5ex>[ddddl] \\ \\ 
      & f \ar[ddl] \ar[ddr] &  \\ \\ 
      a  \ar@<-.5ex>@/_/[uur] \ar@<1ex>@/^/[uuuur] & & b \\ }
 & \qquad\qquad\qquad &
 \xymatrix@R=.2pc@C=3pc{\\ 
      & \Gamma  \ar@(ul,ur) \ar@/_/[dd] \ar@/^/[dd] \\ \\ 
      & g \ar[ddl] \ar[ddr] \ar[dd] &  \\ \\ 
      c \ar@<.5ex>@/^/[uuuur] \ar@<1ex>@/^2ex/[uuuur]  
        \ar@<.10ex>@/^/[uur] \ar@<.5ex>@/^2ex/[uur] 
        \ar@<-.10ex>@/_/[uur] \ar@<-.5ex>@/_2ex/[uur] & 
      d & 
      e  \ar@<.10ex>@/^/[uul] \ar@<.5ex>@/^2ex/[uul] \ar@<-.5ex>@/_/[uuuul] } \\
 \cline{1-1} \cline{3-3}  
  \end{array}   $$
Then $G$ rewrites into $H$ using the rule~$\mu$  
and the matching $L\to G$ defined by the inclusion. Indeed, 
as specified by the cloning multiplicities, 
the edge going out of node $a$ towards $\Gamma$ is cloned three times,
two times by the edges going out from $c$ towards $\Gamma$ 
($C^+(a,c)=2$) and a third time by the edge
going out from $e$ ($C^+(a,e)=1$), the node $b$ is erased as well as
all its incident edges, and the incoming edges of $f$ are duplicated 
($C^-(f,g)=2$) and redirected towards $g$.
The edge from $a$ towards $f$ is copied four times ($C^+(a,c) \times C^-(f,g)= 4$) 
from $c$ to $g$ and two times ($C^+(a,e) \times C^-(f,g)= 2$) from $e$ to~$g$.
\end{exmp}

%%%%%%%%%%%%%%%%%%%%%%%%%%%%%%%%%%%%%%%%%%%%%%%
\subsection{An algebraic definition}  
\label{subsec:alg}

In this section we define the polarized sesqui-pushout 
rewriting system (\psqpo). In Section~\ref{subsec:equiv}
we will prove that \psqpo and AlgoPC (from Section~\ref{subsec:set}) 
are equivalent. 
The \psqpo rewriting can be seen as a generalization of 
the sesqui-pushout approach \cite{CorradiniHHK06}, 
applied to a heterogeneous categorical framework including both
``ordinary'' graphs and polarized graphs,
which allows more flexibility in cloning. 
The definition of the \psqpo rewriting relies on the well-known 
categorical notions of pushout (PO) and pullback (PB).
Pushouts of graphs are described in Proposition~\ref{prop:alg-po} 
and final pullback complements of polarized graphs 
in Proposition~\ref{prop:alg-fpbc}.
These results are proved in the Appendix
(Propositions~\ref{app-prop:gra-po} and~\ref{app-prop:pol-pc}, respectively). 

\begin{prop}
\label{prop:alg-po} 
Let $r:K\to R$ be a morphism of graphs and $d:K\to D$ a matching of graphs.
The following square, where $h$ is the inclusion, is a pushout of $d$ and $r$ in $\Gr$.

  $$ \xymatrix@C=6pc{ 
  K \ar[d]_{d} \ar[r]^{r} & R \ar[d]^{h} \\ 
  D=(K+\ol{K})+_e\ti{K} \ar[r]_{r_1=(r+\id_{\ol{K}})+_e\ti{r}} & 
    H=(R+\ol{K})+_e \ti{R}   
  % PO
    \ar@{-}[]+L+<-6pt,+1pt>;[]+LU+<-6pt,+6pt> 
    \ar@{-}[]+U+<-1pt,+6pt>;[]+LU+<-6pt,+6pt> 
  \\ 
  }$$
where:
  $ \edgnp{\ti{R}}{n}{p} = 
  \sum_{n_D\in r_1^{-1}(n),p_D\in r_1^{-1}(p)} \edgnp{\ti{K}}{n_D}{p_D}
  \;\mbox{ for all }\; n,p\in\nod{H}  $
  and $ \ti{r}:\ti{K}\to\ti{R} 
  \mbox{ maps } n_D\rupto{e} p_D \mbox{ to } r_1(n_D)\rupto{e} r_1(p_D) \;.$
\end{prop} 

This means that $\nod{H}=\nod{R}+\nod{\ol{K}}$ and 
$\edg{H}=(\edg{R}+\edg{\ol{K}})+_e 
\sum_{n,p\in\nod{R}+\nod{\ol{K}}} \edgnp{\ti{R}}{n}{p}$
where: 
  $$ \edgnp{\ti{R}}{n}{p} = \begin{cases}
   \emptyset 
  & \;\mbox{ when }\; n,p \in\nod{\ol{K}}  \\
   \sum_{p_K\in r^{-1}(p)} \edgnp{\ti{K}}{n}{p_K}
  & \;\mbox{ when }\; n\in\nod{\ol{K}}, \;p \in\nod{R}  \\
  \sum_{n_K\in r^{-1}(n)} \edgnp{\ti{K}}{n_K}{p}
  & \;\mbox{ when }\; n \in\nod{R},\; p\in\nod{\ol{K}}  \\
  \sum_{n_K\in r^{-1}(n),p_K\in r^{-1}(p)} \edgnp{\ti{K}}{n_K}{p_K}
  & \;\mbox{ when }\; n,p \in\nod{R}  \\ 
  \end{cases} $$ 

Let us now define final pullback complements in the naive way, 
this definition coincides with the one in \cite{DyckThol87,CorradiniHHK06} 
when both exist.

\begin{defn}
\label{defi:alg-fpbc}
In a category $\catC$,
let $a:X\to Y$ and $g:Y\to Y_1$ be consecutive morphisms.
A \emph{pullback complement} (PBC) of $a$ and $g$ 
is an object $X_1$ with a pair of morphisms $f:X\to X_1$, $a_1:X_1\to Y_1$ 
such that there is a pullback:
$$  \xymatrix@C=4pc@R=1.5pc{
  Y \ar[d]_{g} & X \ar[l]_{a} \ar[d]^{f} 
    % PB
    \ar@{-}[]+L+<-6pt,-1pt>;[]+LD+<-6pt,-6pt> 
    \ar@{-}[]+D+<-1pt,-6pt>;[]+LD+<-6pt,-6pt>  
  \\
  Y_1 & X_1 \ar[l]_{a_1} 
  \\
  } $$
A morphism $k:(X_1,f,a_1)\to (X'_1,f',a'_1)$ of pullback complements of $a$ and $g$ 
is a morphism $k:X_1\to X'_1$ in $\catC$ 
such that $k\circ f=f'$ and $a'_1\circ k=a_1$. 
This yields the category of pullback complements of $a$ and $g$,
and the \emph{final pullback complement} (FPBC) of $a$ and $g$ 
is defined as the final object in this category, 
if it does exist.
The FPBC is illustrated as follows: 
$$  \xymatrix@C=4pc@R=1.5pc{
  Y \ar[d]_{g} & X \ar[l]_{a} \ar[d]^{f} 
    % PB
    \ar@{-}[]+L+<-6pt,-1pt>;[]+LD+<-6pt,-6pt> 
    \ar@{-}[]+D+<-1pt,-6pt>;[]+LD+<-6pt,-6pt>  
  \\
  Y_1 & X_1 \ar[l]_{a_1} 
  % FPBC
    \ar@{-}[]+L+<-6pt,+1pt>;[]+LU+<-6pt,+6pt> 
    \ar@{-}[]+L+<-8pt,+1pt>;[]+LU+<-8pt,+8pt> 
    \ar@{-}[]+U+<-1pt,+6pt>;[]+LU+<-6pt,+6pt>  
    \ar@{-}[]+U+<-1pt,+8pt>;[]+LU+<-8pt,+8pt> 
  \\
  } $$
\end{defn}

\begin{prop}
\label{prop:alg-fpbc} 
Let $l:\grpol{K} \to \grpol{L}$ be a morphism and 
$m:\grpol{L}\to \grpol{G}$ a matching of polarized graphs.
The following square, where $d$ is the inclusion, is a FPBC of $l$ and $m$ in 
$\Grpol$: 
  $$ \xymatrix@C=6pc{ 
  \grpol{L} \ar[d]_{m} & \grpol{K} \ar[l]_{l} \ar[d]^{d} 
    % PB
    \ar@{-}[]+L+<-6pt,-1pt>;[]+LD+<-6pt,-6pt> 
    \ar@{-}[]+D+<-1pt,-6pt>;[]+LD+<-6pt,-6pt>  
 \\
   \grpol{G}=(\grpol{L}+\ol{\grpol{L}})+_e\ti{\grpol{L}} & 
     \grpol{D}=(\grpol{K}+\ol{\grpol{L}})+_e\ti{\grpol{K}} \ar[l]^{l_1=(l+\id_{\ol{\grpol{L}}})+_e\ti{l}} 
  % FPBC
    \ar@{-}[]+L+<-6pt,+1pt>;[]+LU+<-6pt,+6pt> 
    \ar@{-}[]+L+<-8pt,+1pt>;[]+LU+<-8pt,+8pt> 
    \ar@{-}[]+U+<-1pt,+6pt>;[]+LU+<-6pt,+6pt>  
    \ar@{-}[]+U+<-1pt,+8pt>;[]+LU+<-8pt,+8pt> 
  \\
  }$$
where:
  $ \edgnp{\ti{K}}{n_D}{p_D} = \edgnp{\ti{L}}{l_1(n_D)}{ l_1(p_D)} 
   \;\mbox{ for all }\; n_D\in\nod{D}^+, \; p_D\in\nod{D}^-$ 
   and otherwise $\edgnp{\ti{K}}{n_D}{p_D} = \emptyset) $
  and $\ti{l}:\ti{K}\to\ti{L} 
  \mbox{ maps } n_D\rupto{e} p_D \mbox{ to } l_1(n_D)\rupto{e} l_1(p_D) \;.$
\end{prop}

\begin{defn}
\label{defi:alg-pcpo}
  A \emph{\psqpo rewrite rule} $\rho$ is a polarized graph $\grpol{K}=(K,K^\pm)$
  together with a span of graphs $\spa{L}{l}{K}{r}{R}$. 
  This is denoted:
  $ \xymatrix@C=2pc{
  L & \grpol{K} \ar[l]_{l} \ar[r]^{r} & R \\ } $

  \noindent
  Let $\rho=\spa{L}{l}{\grpol{K}}{r}{R}$ be a \psqpo rewrite rule, 
  $G$ a graph and $m:L \to G$ a matching of graphs. 
  The \emph{\psqpo rewrite step} applying the rule $\rho$
  to the matching $m$ builds the graph $H$ and the matching $h:R\to H$
  such that $h$ is the inclusion, in four steps: 
  \begin{enumerate}
  \item The first step builds the polarized graphs
    $\grpol{L}=\polar(L)$ induced by $L$ and $\grpol{G}=\polar(G)$ 
    induced by $G$ (see Definition \ref{defi:alg-minmax}).
  \item Then a FPBC of $m$ and $l$ as in Proposition~\ref{prop:alg-fpbc} $\Grpol$ is built, 
    which gives rise to a polarized graph $\grpol{D}$, 
    a morphism $l_1:\grpol{D}\to\grpol{G}$ 
    and a matching $d:\grpol{K}\to\grpol{D}$ in $\Grpol$.
  \item The polarizations are forgotten, 
    by considering the underlying graphs $K=\depol(\grpol{K})$ 
    and $D=\depol(\grpol{D})$.
  \item A pushout of $d$ and $r$ in $\Gr$ is constructed 
    as in Proposition~\ref{prop:alg-po},
    which gives rise to a graph $H$, 
    a morphism $r_1:D\to H$ 
    and a matching $h:R\to H$ in~$\Gr$. 
  \end{enumerate}
A \psqpo rewrite step is written $\step{G}{\arpsqpo}{H}$ or $\stepp{m}{\arpsqpo}{\rho}{h}$  
and it is illustrated by the following diagram: 
$$ 
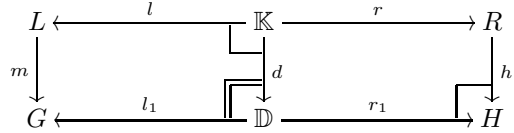
\xymatrix@C=6pc{
L \ar[d]_{m} & \grpol{K} \ar[l]_{l} \ar[r]^{r} \ar[d]^{d} 
    % PB
    \ar@{-}[]+L+<-6pt,-1pt>;[]+LD+<-6pt,-6pt> 
    \ar@{-}[]+D+<-1pt,-6pt>;[]+LD+<-6pt,-6pt>  
& R \ar[d]^{h}\\
G & \grpol{D}  \ar[l]_{l_1} \ar[r]^{r_1} 
  % FPBC
    \ar@{-}[]+L+<-6pt,+1pt>;[]+LU+<-6pt,+6pt> 
    \ar@{-}[]+L+<-8pt,+1pt>;[]+LU+<-8pt,+8pt> 
    \ar@{-}[]+U+<-1pt,+6pt>;[]+LU+<-6pt,+6pt> 
    \ar@{-}[]+U+<-1pt,+8pt>;[]+LU+<-8pt,+8pt> 
& H 
  % PO
    \ar@{-}[]+L+<-6pt,+1pt>;[]+LU+<-6pt,+6pt> 
    \ar@{-}[]+U+<-1pt,+6pt>;[]+LU+<-6pt,+6pt> 
\\
}$$
\end{defn}

It follows that a \psqpo rewrite step builds a rule
$\rho_1=(\spa{G}{l_1}{\grpol{D}}{r_1}{H})$. 
Merging Propositions~\ref{prop:alg-po} and~\ref{prop:alg-fpbc} 
yields the following result,
which provides an explicit description of \psqpo rewrite steps. 

\begin{thm}
\label{theo:alg-step}
Let $\rho=(\spa{L}{l}{\grpol{K}}{r}{R})$ be a \psqpo rewrite rule 
and $m : L \to G$ a matching, so that $G = (L+\ol{L})+_e\ti{L}$.
The \psqpo rewrite step applying $\rho$ to $m$ builds the matching
$h:R\to H$ where $h$ is the inclusion and 
  $ H=(R+\ol{L})+_e\ti{R} $
where, for all nodes $n,p$ in $\nod{H}$: 
  $$ \edgnp{\ti{R}}{n}{p} = \begin{cases}
  \sum_{n_K^+\in r^{-1}(n),p_K^-\in r^{-1}(p)} \edgnp{\ti{L}}{l(n_K)}{l(p_K)}
  & \;\mbox{ when }\; n,p \in\nod{R}  \\
  \sum_{n_K^+\in r^{-1}(n)} \edgnp{\ti{L}}{l(n_K)}{p}
  & \;\mbox{ when }\; n \in\nod{R},\; p\in\nod{\ol{L}}  \\
  \sum_{p_K^-\in r^{-1}(p)} \edgnp{\ti{L}}{n}{l(p_K)}
  & \;\mbox{ when }\; n\in\nod{\ol{L}}, \;p \in\nod{R}  \\
  \emptyset 
  & \;\mbox{ when }\; n,p \in\nod{\ol{L}}  \\
  \end{cases} $$ 
\end{thm}

This means that $H$ is made of a copy of $R$ together with
the non-matching nodes of $G$ 
(i.e., nodes of $G$ which are not in the image of the matching)  
and with an edge $n \longrupto{(n_D,p_D,e)} p$ 
for each $n_D$ in $\nod{K}^+ + \nod{\ol{L}}$ such that $r_1(n_D)=n$, 
each $p_D$ in $\nod{K}^- + \nod{\ol{L}}$ such that $r_1(p_D)=p$
and each $n_G \rupto{e} p_G$ in $\edg{G}$  
where $n_G=l_1(n_D)$ and $p_G=l_1(p_D)$.
Since $l_1$ and $r_1$ are the identity on $\nod{\ol{L}}$, 
whenever both $n$ and $p$ are in $\nod{\ol{L}}$ then 
$ \edgnp{H}{n}{p} = \edgnp{G}{n}{p}$. 

\begin{exmp}
\label{exam:alg-pcpo} 
Example~\ref{exam:set-setpc} is now seen from the categorical point of view. 
The rewrite rule $\mu$ can be translated to the following rule
 $$ \begin{array}{|c|c|c|c|c|}
 \cline{1-1} \cline{3-3} \cline{5-5}   
 L && \grpol{K} && R \\
 \cline{1-1} \cline{3-3} \cline{5-5}    
 
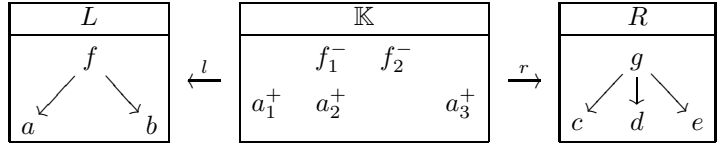
\xymatrix@=1pc{ & f \ar[dl] \ar[dr] &  \\
                  a  &  & b} & 
  \xymatrix@R=.2pc@C=1pc{ \\ 
                  & \ar[l]_{l} } & 
 \xymatrix@R=0pc@C=.5pc{ &f_1^- & f_2^- &  \\
                  a_1^+  & a_2^+ && a_3^+ } &
  \xymatrix@R=.2pc@C=1pc{ \\ 
                  \ar[r]^{r} &} & 
 \xymatrix@=1pc{ & g \ar[dl] \ar[dr] \ar[d] &  \\
                 c  & d & e} \\
 \cline{1-1} \cline{3-3} \cline{5-5}    
  \end{array}   $$
where $l(f_1)=l(f_2)=f$, $l(a_1)=l(a_2)=l(a_3)=a$, $r(f_1)=r(f_2)=g$, 
$r(a_1)=r(a_2)=c$ and $r(a_3)=e$. 
As in Example~\ref{exam:set-setpc}, the matching is the inclusion
of $L$ in $G$ (below) and the \psqpo rewrite step builds: 
$$ \begin{array}{|c|c|c|c|c|}
 \cline{1-1} \cline{3-3} \cline{5-5}   
 G && \grpol{D} && H \\
 \cline{1-1} \cline{3-3} \cline{5-5}    
 
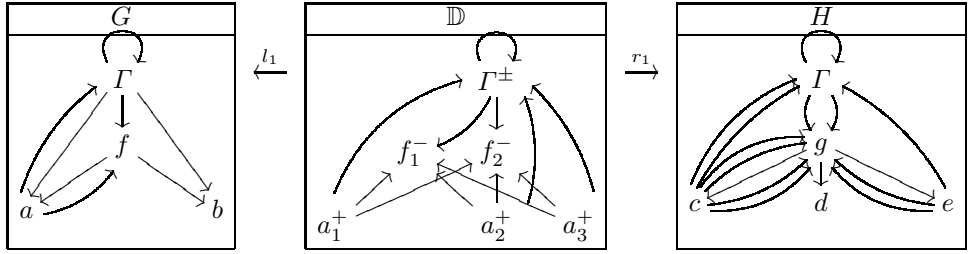
\xymatrix@R=.2pc@C=2pc{ \\  
      & \Gamma \ar@(ul,ur) \ar[dd] \ar@<.5ex>[ddddr] \ar@<-.5ex>[ddddl] \\ \\ 
      & f \ar[ddl] \ar[ddr] &  \\ \\ 
      a  \ar@<-.5ex>@/_/[uur] \ar@<1ex>@/^/[uuuur] & & b \\ } & 
 \xymatrix@R=.2pc@C=1pc{ \\ & \ar[l]_{l_1} \\ } & 
 \xymatrix@R=.2pc@C=1pc{ \\ 
     && \Gamma^\pm \ar@(ul,ur) \ar@/^/[ddl] \ar[dd] \\ \\  
     & f_1^- & f_2^- &  \\ \\
     a_1^+ \ar@<1ex>@/^3ex/[uuuurr] \ar[uur] \ar[uurr]  & & 
     \ar[uul]\ar[uu]a_2^+ \ar@<-2ex>@/_/[uuuu] & 
     \ar[uull]\ar[uul]a_3^+ \ar@<-2ex>@/_/[uuuul] \\ } &
 \xymatrix@R=.2pc@C=1pc{ \\ \ar[r]^{r_1} & \\ } & 
 \xymatrix@R=.2pc@C=3pc{ \\ 
      & \Gamma  \ar@(ul,ur) \ar@/_/[dd] \ar@/^/[dd] \\ \\ 
      & g \ar[ddl] \ar[ddr] \ar[dd] &  \\ \\ 
      c \ar@<.5ex>@/^/[uuuur] \ar@<1ex>@/^2ex/[uuuur]  
        \ar@<.10ex>@/^/[uur] \ar@<.5ex>@/^2ex/[uur] 
        \ar@<-.10ex>@/_/[uur] \ar@<-.5ex>@/_2ex/[uur] & 
      d & 
      e  \ar@<.10ex>@/^/[uul] \ar@<.5ex>@/^2ex/[uul] \ar@<-.5ex>@/_/[uuuul] } \\
 \cline{1-1} \cline{3-3} \cline{5-5}    
 \end{array}   $$
The resulting graph $H$ and matching $h:R\to H$
are the same as in Example~\ref{exam:set-setpc}.
\end{exmp}

%%%%%%%%%%%%%%%%%%%%%%%%%%%%%%%%%%%%%%%%%%%%%%%
\subsection{Equivalence of the algorithmic and the algebraic definitions}
\label{subsec:equiv}

Now we can prove that the categorical
and the algorithmic definitions given in the previous Sections~\ref{subsec:alg} 
and~\ref{subsec:set} are equivalent. 
We start by noticing that edges and neutral nodes in the interface $\grpol{K}$
of a rule  $\rho=(\spa{L}{l}{\grpol{K}}{r}{R})$ can be removed without changing the operational semantics of the rule.

\begin{defn}
\label{defi:alg-norm}
For each \psqpo rewrite rule $\rho=(\spa{L}{l}{\grpol{K}}{r}{R})$,  
the \emph{normalization} of $\rho$ is the 
\psqpo rewrite rule $\norm(\rho)=(\spa{L}{l'}{\grpol{K'}}{r'}{R})$  
where $\grpol{K'}$ is 
obtained by removing all the edges and all the neutral nodes from $\grpol{K}$ 
and by replacing every node $n^\pm$ by two nodes  
${n_+}^+$ and ${n_-}^-$, 
and where $l'$ and $r'$ are similar to $l$ and $r$ 
with $l'(n_+)=l'(n_-)=l(n)$ and $r'(n_+)=r'(n_-)=r(n)$.
\end{defn}

\begin{rem}
\label{rem:alg-norm}
The rules $\rho$ and $\norm(\rho)$ are equivalent, 
in the sense that 
for all matchings of graphs $m:L\to G$ and $h:R\to H$:
   $ \stepp{m}{\arpsqpo}{\rho}{h} \;\mbox{ if and only if }\;
\stepp{m}{\arpsqpo}{\norm(\rho)}{h} \;.$ Indeed, one can note that in
Theorem \ref{theo:alg-step} the definition of $H$ does not use 
the neutral nodes nor the edges of $\grpol{K}$, moreover for every node 
$n^{\pm}$ in $\grpol{K}$ (see Definition~\ref{defi:alg-norm}) the nodes ${n_+}^+$ 
and ${n_-}^-$ of $\grpol{K}'$ get identified in $R$ hence in $H$. 
\end{rem}

Hereafter, we provide translations of the rules defined according to
the algorithmic approach into the the algebraic approach and vice-versa.
 
\begin{defn}  
\label{defi:equiv-equiv}
Let $\mu=(L,R,C^+,C^-)$ be an AlgoPC rewrite rule. 
Let $\grpol{K}$ be the polarized graph without edges and  with, 
for each $\star\in\{+,-\}$,  
a node ${(n_L,n_R)_{i,\star}}^\star$
for each pair of nodes $(n_L,n_R) \in \nod{L}\times\nod{R}$
and each $i\in \{1,\dots,C^\star(n_L,n_R)\}$.
Let $\spa{L}{l}{K}{r}{R}$ be the span of graphs 
such that $l((n_L,n_R)_{i,\star})=n_L$ and $r((n_L,n_R)_{i,\star})=n_R$. 
Let $\catversion(\mu)$ denote the \psqpo rewrite rule 
$\spa{L}{l}{\grpol{K}}{r}{R}$.

Let $\rho=(\spa{L}{l}{\grpol{K}}{r}{R})$ be a \psqpo rewrite rule.
For each $\star\in\{+,-\}$ and each $(n_L,n_R) \in \nod{L}\times\nod{R}$, 
let $C^\star(n_L,n_R)$ be the number of nodes $n_K$
in $\nod{K}^\star$ such that $l(n_K)=n_L$ and $r(n_K)=n_R$. 
Let $\setversion(\rho)$ denote the AlgoPC rewrite rule 
$(L,R,C^+,C^-)$.

\end{defn}

\begin{prop} 
\label{prop:equiv-equiv}
Let $\mu$ be an AlgoPC rewrite rule and $\rho$ a \psqpo rewrite rule.
Then $\setversion(\catversion(\mu)) = \mu$ 
and $\catversion(\setversion(\rho)) = \norm(\rho)$. 
In addition, the rules $\mu$ and $\catversion(\mu)$ are equivalent, 
and the rules $\rho$ and $\setversion(\rho)$ are equivalent, 
in the sense that 
for all matchings of graphs $m:L\to G$ and $h:R\to H$:
$ \stepp{m}{algopc}{\mu}{h}$ if and only if 
 $\stepp{m}{\arpsqpo}{\catversion(\mu)}{h}$, and
 $\stepp{m}{\arpsqpo}{\rho}{h}$ if and only if 
 $\stepp{m}{algopc}{\setversion(\rho)}{h}$.
\end{prop} 

The proof consists in a simple comparison of the definition of an 
AlgoPC rewrite step (Definition \ref{defi:set-setpc}) with Theorem 
\ref{theo:alg-step}. 

%%%%%%%%%%%%%%%%%%%%%%%%%%%%%%%%%%%%%%%%%%%%%%%
\subsection{Functoriality of polarized sesqui-pushout rewriting} 
\label{subsec:functo}

In \cite{DuvalEP11} we have defined categorical rewriting systems 
and their (horizontal) composition, 
as well as the functoriality property of a rewriting system
(which corresponds to the vertical composition in \cite{Lowe10}) 
and we have proved that the composition of categorical rewriting systems 
preserves the functoriality property.  
We now show that \psqpo rewriting can be seen as a 
categorical rewriting system  
which can be presented as the composition of
four functorial categorical rewriting systems,
from which it follows that \psqpo rewriting is functorial 
(Proposition~\ref{prop:pcpo-funct}). 

\begin{defn}
\label{defi:crs-defi}
A \emph{categorical rewriting system} is a span of functors:
  $\spa{\catL}{\funL}{\catP}{\funR}{\catR}$
together with a family of partial functions 
$ \funS=(\funS_\rho)_{\rho\in|\catP|} $
indexed by the objects $\rho$ of $\catP$,
such that the partial  function $\funS_\rho$ 
maps morphisms in $\catL$ with source $\funL(\rho)$ 
to morphisms in $\catP$ with source $\rho$, 
is such a way that $\funL(\funS_\rho(f))=f$ for every $f$ in the domain 
of $\funS_\rho$.
The objects of $\catP$ are the \emph{rewrite rules} or \emph{productions},
the morphisms of $\catL$ and $\catR$ are the 
left-hand side and right-hand side \emph{matches},
and the partial function $\funS_\rho$ is the \emph{rewriting process} function 
with respect to $\rho$, with domain denoted as $\dom(\funS_\rho)$. 
Given a rule $\rho$, the \emph{rewrite step applying $\rho$} 
is the partial function 
from the set of morphisms in $\catL$ with source $\funL(\rho)$ 
to the set of morphisms in $\catR$ with source $\funR(\rho)$ 
which maps every match $f$ in $\dom(\funS_\rho)$
to the match $g=\funR(\funS_\rho(f))$. 
\end{defn}

The fact that $\funS_\rho$ is partial allows to deal with conditions
on a matching (with respect to the rule $\rho$) 
ensuring that $\rho$ can be fired:
this is the case with gluing conditions in the double pushout approach
and with conflict-freeness conditions in the sesqui-pushout approach. 

The definitions below are given up to isomorphism, more details can be found 
in \cite{DuvalEP11}. 

\begin{defn}
\label{definition:crs-comp}
Let $\rs=(\spa{\catL}{\funL}{\catP}{\funR}{\catR},\funS)$
and $\rs'=(\spa{\catL'}{\funL'}{\catP'}{\funR'}{\catR'},\funS')$
be two categorical rewriting systems
which are consecutive, in the sense that $\catR=\catL'$. 
The \emph{composition} of $\rs$ and $\rs'$ is the categorical 
rewriting system  
 $ \rs' \circ \rs = (\spa{\catL}{\funL''}{\catP''}{\funR''}{\catR'},
\funS''_{(\rho,\rho')})$
where $\spa{\catL}{\funL''}{\catP''}{\funR''}{\catR'}$ is the composition
of the spans in $\rs$ and $\rs'$ and where the family of partial functions 
$\funS''=(\funS''_{\rho''})_{\rho''\in|\catP''|}$ is defined as follows, 
for each $\rho''=(\rho,\rho')$ in $\catP''$:
the domain of $\funS''_{\rho''}$ is made of the 
morphisms $f$ in $\dom(\funS_\rho)$ such that 
$\funR(\funS_\rho(f))$ is in $\dom(\funS'_{\rho'})$, and for
each $f\in\dom(\funS''_{\rho''})$:   
    $ \funS''_{(\rho,\rho')}(f) = (\funS_\rho(f), \funS'_{\rho'}(f'))$
    where  $f'= \funR(\funS_\rho(f)) \;.$
$$ 
\xymatrix@C=2pc@R=1.5pc{
L \ar[d]_{f} \ar@{~>}[rr]^{\rho} & \ar@{}[d]|{\funS_\rho(f)} & 
R=L' \ar[d]^{f'} \ar@{~>}[rr]^{\rho'} & \ar@{}[d]|{\funS_{\rho'}(f')} & 
R' \ar[d]^{f''}\\ 
L_1 \ar@{~>}[rr]_{\rho_1} && 
R_1=L'_1 \ar@{~>}[rr]_{\rho'_1} && 
R'_1 \\ 
}  \quad
\xymatrix@R=.8pc{  \\ = \\}  \quad
 \xymatrix@C=2pc@R=1.5pc{
L \ar[d]_{f} \ar@{~>}[rr]^{\rho''} & \ar@{}[d]|{\funS_{(\rho,\rho')}(f)} & R' \ar[d]^{f''} \\ 
L_1 \ar@{~>}[rr]_{\rho''_1} && R'_1 \\ 
}$$
\end{defn}

\begin{defn}
\label{definition:crs-fun} 
A categorical rewriting system
$(\spa{\catL}{\funL}{\catP}{\funR}{\catR},\funS)$ is \emph{functorial}
if for each rule $\rho:L\arsto R$ the partial function $\funS_\rho$
satisfies: the identity $\id_L$ is in the domain of $\funS_\rho$ and $
\funS_\rho(\id_L)=\id_\rho $, and for each pair of consecutive
morphisms $f_1:L\to L_1$ and $f_2:L_1\to L_2$ in $\catL$, if
$f_1\in\dom(\funS_\rho)$ and $f_2\in\dom(\funS_{\rho_1})$, where
$\rho_1$ denotes the target of $\funS_\rho(f_1)\,$, then $f_2\circ
f_1\in\dom(\funS_\rho)$ and $ \funS_\rho(f_2\circ f_1) =
\funS_{\rho_1}(f_2)\circ \funS_\rho(f_1) $.
\end{defn}

\begin{exmp}
\label{exam:crs-func} 
Every functor $\Theta:\catL\to\catR$ gives rise to 
a categorical rewriting system 
$\rs_{\Theta} = (\spa{\catL}{\id}{\catL}{\Theta}{\catR},\funS)$ 
where $\funS$ is made of the identity functions. 
The categorical rewriting system $\rs_{\Theta}$ is functorial. 
It is called \emph{the rewriting system associated to~$\Theta$}.
\end{exmp}

The \psqpo rewriting system may be recovered by composing 
four categorical rewriting systems.
Two of them are instances of Example~\ref{exam:crs-func},
the other two are defined in \cite{DuvalEP11}. 
\begin{enumerate}
\item The rewriting system $\rs_{\polar}$ associated to the functor 
$\polar:\Gr\to\Grpol$, which maps each graph 
to its induced polarized graph.
\item The FPBC rewriting system $\rs_{\mathrm{fpbc},\Grpol}$ 
on the category $\Grpol$, which is defined as follows.
A rule is a morphism of polarized graphs $l:\grpol{L}\lto\grpol{K}$
and a morphism from $l:\grpol{L}\lto\grpol{K}$ 
to $l_1:\grpol{G}\lto\grpol{D}$  
is made of two matchings of polarized graphs 
$m:\grpol{L}\to \grpol{G}$ and $d:\grpol{K}\to \grpol{D}$ such that 
$l_1\circ d= m\circ l$. 
For every rule $l:\grpol{L}\lto\grpol{K}$ 
and matching $m:\grpol{L}\to \grpol{G}$, 
the morphism $\funS_l(m)$ is built from the final pullback complement  
of $m$ and $l$ in the category $\Grpol$.   
\item The rewriting system $\rs_{\depol}$ associated to the functor 
$\depol:\Grpol\to\Gr$, which maps each polarized graph 
to its underlying graph.
\item The PO rewriting system $\rs_{\mathrm{po},\Gr}$ 
on the category $\Gr$, which is defined as follows. 
A rule is a morphism of graphs $r:K\to R$ 
and a morphism from $r:K\to R$ to $r_1:D\to H$ 
is made of two matchings of graphs 
$d:K\to D$ and $h:R\to H$ such that 
$r_1\circ d= h\circ r$.
For every rule $r:K\to R$ and matching $d:K\to D$, 
the morphism $\funS_r(d)$ is built from the pushout 
of $d$ and $r$ in the category $\Gr$.  
\end{enumerate}

\begin{defn} 
\label{defi:pcpo-crs}
The \emph{\psqpo rewriting system} $\rs_{\mathrm{\psqpo}}$
is the categorical rewriting system 
  $ \rs_{\mathrm{\psqpo}} =  
  \rs_{\mathrm{po},\Gr} \circ \rs_{\depol} \circ  
  \rs_{\mathrm{fpbc},\Grpol} \circ \rs_{\polar} $  
\end{defn} 

\begin{prop}
\label{prop:pcpo-funct}
The \psqpo rewriting system is functorial. 
\end{prop}

%%%%%%%%%%%%%%%%%%%%%%%%%%%%%%%%%%%%%%%%%%%%%%%
\section{An extension to labeled polarized graphs}  
\label{sec:lab} 

For several modeling purposes, it is useful to add labels to nodes and edges. 
In this section we discuss an extension of our proposal 
in order to perform polarized sesqui-pushout graph transformation 
on labeled graphs. 
We provide syntactic conditions which ensure the existence of the
constructions involved in the rewriting process.
Hereafter, two sets $\lset_N$ and $\lset_E$ are given, 
they are called the set of \emph{labels} for nodes and for edges, respectively. Moreover, 
all the constructions are considered up to isomorphism. 

\begin{defn}
\label{defi:lab-cat}
A \emph{labeled graph} $(X,\lfun)$ is a graph $X$ together with two partial functions 
$\lfun:\nod{X} \parto \lset_N$ for the \emph{labeling} of nodes 
and $\lfun:\edg{X} \parto \lset_E$ for the \emph{labeling} of edges. 
A \emph{morphism} of labeled graphs $f:(X,\lfun_X)\to (Y,\lfun_Y)$ 
is a morphism of graphs $f:X\to Y$ which preserves the labels, 
in the sense that if a node or an edge $x$ in $X$ is labeled with $a$
then $f(x)$ in $Y$ is labeled with $a$
(if $x$ is unlabeled there is no restriction on the labeling of $f(x)$). 
This provides the category $\Grlab$ of labeled graphs
(with labels in $\lset_N$ and $\lset_E$). 
\end{defn}

A labeled graph $(X,\lfun)$ is often simply denoted $X$. 
A node $x$ is denoted $x\colon a$ if it is labeled with $a$ 
and $x \colon \nolab$ if it is unlabeled.  
An edge $x \to y$ is denoted $x \rupto{a} y$ if it is labeled with $a$ 
and simply $x \to y$ if it is unlabeled.  
A \emph{matching} of labeled graphs is a matching of graphs 
which preserves the labels. 
Proposition~\ref{prop:alg-po} is generalized to labeled graphs as follows.

\begin{prop} 
\label{prop:lab-po} 
Let $r:K\to R$ be a morphism of labeled graphs 
and $d:K\to D$ a matching of labeled graphs.
Let us assume that:
\begin{itemize}
\item For each node or edge $x$ in $K$,
if $r(x):a$ and $d(x):b$, then $a=b$.
\item For each distinct nodes or edges $x,y$ in $K$,
if $r(x)=r(y)$, $d(x):a$ and $d(y):b$, then $a=b$.
\end{itemize} 
Then the pushout of $d$ and $r$ in $\Grlab$ exists,
its underlying diagram of graphs is the pushout of $d$ and $r$ in $\Gr$
and each node or edge $x$ in $H$ is labeled if and only if 
it is the image of a labeled node or edge in $R$ or in $D$.
\end{prop} 

Thanks to the assumptions, no conflict may arise 
when labeling the graph $H$: 
if a node or edge $x$ in $H$ is the image of 
several nodes or edges in $R$ or in $D$
(at most one in $R$ and maybe several in $D$), 
then all of them have the same label, which becomes the label of $x$. 

Since polarizations and labelings do not interfere, 
these definitions and results are easily combined 
with the definitions and results in Section~\ref{subsec:pol}. 
This provides the category $\Grpolab$ of \emph{labeled polarized graphs},
and Proposition~\ref{prop:alg-fpbc} is generalized 
to labeled polarized graphs as follows.

\begin{prop} 
\label{prop:lab-fpbc} 
Let $l:\grpol{K} \to \grpol{L}$ be a morphism of labeled polarized graphs and 
$m:\grpol{L}\to \grpol{G}$ a matching of labeled polarized graphs.
Then the FPBC of $l$ and $m$ exists,
its underlying diagram of graphs is the FPBC in $\Gr$
and each node or edge $x_D$ in the graph $D$ is labeled as follows:
if $x_D$ is not in the image of $K$ then $x_D$ is labeled in $D$ like $l_1(x_D)$ in $G$, 
otherwise $x_D=d(x_K)$ for a unique $x_K$ in $K$ and the label of $x_D$ in $D$
is determined by the labels of $x_K$ in $K$, $x_L=l(x_K)$ in $L$ 
and $x_G=m(x_L)$ in $G$ according to the following labeling patterns:

 $$  \xymatrix@=1pc{ 
  x_L\colon a \ar@{|->}[d] & x_K \colon a\ar@{|->}[d] \ar@{|->}[l] \\ 
  x_G \colon a & x_D \colon a \ar@{|->}[l] 
  \\ 
  }
\quad 
 \xymatrix@=1pc{ 
  x_L\colon \nolab \ar@{|->}[d] & x_K \colon \nolab \ar@{|->}[d] \ar@{|->}[l] \\ 
  x_G \colon a & x_D \colon a \ar@{|->}[l] 
  \\ 
  }
\quad 
 \xymatrix@=1pc{ 
  x_L\colon \nolab \ar@{|->}[d] & x_K \colon \nolab \ar@{|->}[d] \ar@{|->}[l] \\ 
  x_G \colon \nolab & x_D \colon \nolab \ar@{|->}[l] 
  \\ 
  }
\quad 
 \xymatrix@=1pc{ 
  x_L\colon a \ar@{|->}[d] & x_K \colon \nolab \ar@{|->}[d] \ar@{|->}[l] \\ 
  x_G \colon a & x_D \colon \nolab \ar@{|->}[l] 
  \\ 
  }
$$
\end{prop} 

The labeled \psqpo rewrite rules cannot be defined simply as 
\psqpo rewrite rules 
where the graphs are labeled and the morphisms preserve the labels: 
indeed, in order to avoid conflicts in labeling the pushout, 
the assumptions in Proposition~\ref{prop:lab-po}
must be satisfied after the construction of the polarized FPBC
(Proposition~\ref{prop:lab-fpbc}). 
This leads to the following definition. 

\begin{defn}
\label{defi:lab-prod}
A \emph{labeled \psqpo rewrite rule} is a \psqpo rewrite rule 
$\spa{L}{l}{\grpol{K}}{r}{R}$ (Definition~\ref{defi:alg-pcpo}) 
where the graphs are labeled and the morphisms preserve the labels,
 such that the following conditions are fulfilled~: 
\begin{itemize}
\item For each unlabeled node or edge $x$ in $K$,
if $l(x)$ is unlabeled in $L$ then $r(x)$ is unlabeled in $R$.
\item For each distinct unlabeled nodes or edges $x,y$ in $K$,
if $l(x)$ and $l(y)$ are distinct and
at least one of $l(x)$ or $l(y)$ is unlabeled in $L$ then 
$r(x)\ne r(y)$ in $R$. 
\end{itemize}
\end{defn}

\begin{rem}
The \emph{labeled \psqpo rewriting system} 
can be defined similarly to the \psqpo rewriting system 
in Definition~\ref{defi:pcpo-crs}. 
It is a categorical rewriting system which is functorial. 
\end{rem}

\begin{rem}
\label{rem:lab-norm}
When dealing with labeled graphs, Remark~\ref{rem:alg-norm} on the
equivalence of $\rho$ and $\norm(\rho)$ does not hold anymore: now the
edges and the neutral nodes of $\grpol{K}$ cannot be dropped, in
general, without modifying the rewrite step.  Typically neutral nodes
in $\grpol{K}$ can be used for transfering a label from $G$ to $H$,
when this label is not provided by the matching.
\end{rem}

\begin{exmp}
\label{exam:lab-if}
The behavior of the ``{\tt if} $b$ {\tt then\dots else\dots}'' operator 
in imperative languages 
can be modelled thanks to two polarized \psqpo rewrite rules, 
one when $b$ is {\tt true} and another one when $b$ is {\tt false}. 
Here is a possible choice when $b$ is {\tt true}
(morphisms are represented via node name sharing,
for instance $r(m)=r(p)=p,m$ and $l(m)=m$): 
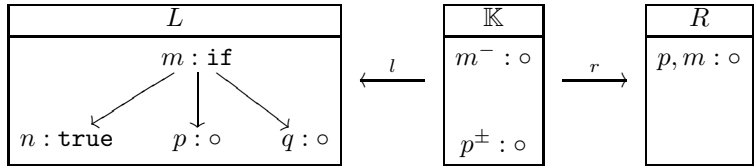
 $$ \begin{array}{|c|c|c|c|c|}
 \cline{1-1} \cline{3-3} \cline{5-5}   
 L && \grpol{K} && R \\
 \cline{1-1} \cline{3-3} \cline{5-5}    
   \xymatrix@R=1.5pc@C=1pc{& m : {\tt if} \ar[d] \ar[dl] \ar[dr] & \\
    n : {\tt true} & p : \nolab &  q : \nolab} & 
  \xymatrix@R=.2pc@C=2pc{ \\ 
                  & \ar[l]_{l} } & 
   \xymatrix@R=1.5pc@C=1pc{ m^- : \nolab  \\
     p^\pm : \nolab } & 
  \xymatrix@R=.2pc@C=2pc{ \\ 
                  \ar[r]^{r} &} & 
   \xymatrix@R=1.5pc@C=1pc{p,m : \nolab} \\
 \cline{1-1} \cline{3-3} \cline{5-5}    
  \end{array}   $$
These rules for modeling ``{\tt if\dots then\dots else\dots}''
are destructive, in the sense that
nodes $n$ and $q$ disappear during the rewrite step. 
Non-destructive rules can also be chosen, here is such a rule for {\tt true}. 

 $$ \begin{array}{|c|c|c|c|c|}
 \cline{1-1} \cline{3-3} \cline{5-5}   
 L && \grpol{K} && R \\
 \cline{1-1} \cline{3-3} \cline{5-5}    
     
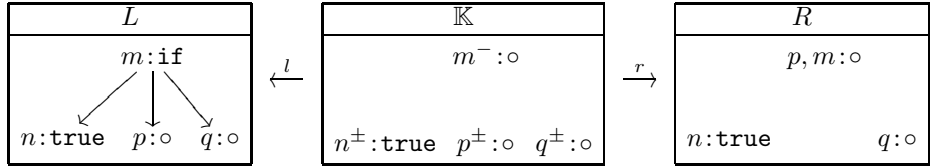
\xymatrix@R=1.5pc@C=0pc{& m \!:\! {\tt if} \ar[d] \ar[dl] \ar[dr] & \\
                                    n \!:\! {\tt true} & p \!:\!
                                    \nolab &  q \!:\! \nolab} &
              \xymatrix@R=.2pc@C=1pc{ \\ & \ar[l]_{l}}  & 
  \xymatrix@R=1.5pc@C=0pc{& m^- \!:\! \nolab & \\
                        n^\pm \!:\! {\tt true} & p^\pm \!:\! \nolab 
                        &  q^\pm \!:\! \nolab } &
  \xymatrix@R=.2pc@C=1pc{ \\ \ar[r]^{r} &} & 
  \xymatrix@R=1.5pc@C=0pc{& p,m \!:\! \nolab & \\
                        n \!:\! {\tt true} & 
                        &  q \!:\! \nolab }\\
 \cline{1-1} \cline{3-3} \cline{5-5}    
  \end{array}   $$

\end{exmp}

%%%%%%%%%%%%%%%%%%%%%%%%%%%%%%%%%%%%%%%%%%%%%%%
\section{Related work} 
\label{sec:rel}

Polarized sesqui-pushout graph rewriting (\psqpo) is a 
new way to perfom graph transformations which offers
different possibilities to clone nodes and their incident edges, in
addition to classical graph transformations (addition and deletion of nodes
and edges). In this section the \psqpo approach is compared 
with other approaches for graph transformations.
 
%%%%%%%%%%%%%%%%%%%%%%%%%%%%%%%%%%%%%%%%%%%%%%%
In \cite{DuvalEP09} an algebraic approach
of termgraph transformation, based on heterogeneous
pushouts (HPO),  has been proposed. There,
a rule is defined as a tuple $(L,R,\tau, \sigma)$ such that $L$ and
$R$ are termgraphs representing the left-hand and the right-hand sides
of the rule, $\tau$ is a mapping from the nodes of $L$ to those of $R$
and $\sigma$ is a partial function from nodes of $R$ to nodes of $L$.
The mapping $\tau$ describes how incoming edges of the nodes in $L$
are connected in $R$ (i.e., global redirection of incoming pointers),
$\tau$ is not required to be a graph morphism as in classical
algebraic approaches of graph transformation.  The role of $\sigma$ is
to indicate the parts of $L$ to be cloned.  These two functions $\tau$
and $\sigma$ have been generalized in our present approach to a span
$\spa{L}{l}{\grpol{K}}{r}{R}$ where the polarized graph 
$\grpol{K}$ is annotated with cloning indications.  Handling
termgraphs as in \cite{DuvalEP09} requires some care
to ensure the preservation of the arity (the number of outgoing edges) of a node
during a transformation process. This requirement prevents from
deletion of nodes and their incident edges in general. To ensure
preservation of node arities, the function $\tau$ is required to be
total. The problem of arity preserving does not appear in
graphs. Thus, in our context, a node, $n$ in $L$, can
actually be deleted (zero clone) with all its incident edges if, for
instance, $n$ has no antecedent in $\grpol{K}$.
With respect to cloning abilities, the HPO approach offers the
possibility to make one or more copies of a node together with its
outgoing edges. Therefore, this way of cloning nodes is limited to
the outgoing edges only and contrasts with the flexible possibilities
of cloning edges proposed in the present paper. 
In fact, whenever a graph $G$
rewrites into $H$ according to the HPO approach using a rule
$(L,R,\tau, \sigma)$, the graph $G$ can also be rewritten into $H$ according
to a rule $\spa{L}{l}{\grpol{K}}{r}{R}$ where morphisms $l$ and $r$ 
encode the functions $\tau$ and $\sigma$ as described below.

\begin{prop}
\label{prop:rel-hpo}
Let $\rho$ be an HPO rule $(L,R,\tau, \sigma)$. 
Then $L$ and $R$ are graphs,
$\tau : |L| \to |R|$ is a total function and 
$\sigma : |R| \to |L|$ is partial function. 
Let $C$ denote the domain of $\sigma$, seen as a graph without edges. 
Let us assume that every node in $C$ has no outgoing edges in $R$ 
(such nodes are kinds of variables).
Let $\rho$ be the rule $\spa{L}{l}{\grpol{K}}{r}{R}$ defined as follows 
\begin{itemize}
\item $K$ is a graph without edges and $|K| = |L| + |C|$,
\item let $n$ be in $|L|$, then $l(n) = n$ and $r(n) = \tau(n)$, 
\item let $n$ be in $|C|$, then $l(n) = \sigma(n)$ and $r(n) = n$,
\item $|K|^+ = |C|$:
  nodes in $C$ are dedicated to make clones with outgoing edges only,  
\item $|K|^- = |L|$:
   nodes of $L$ undergo global redirection of incoming pointers.
\end{itemize}
Then, for an injective matching $m : L \to G$, 
$\step{G}{hpo}{H}$ implies $\step{G}{\arpsqpo}{H}$. 
\end{prop}

%%%%%%%%%%%%%%%%%%%%%%%%%%%%%%%%%%%%%%%%%%%%%%%

Cloning is also one of the features of the sesqui-pushout approach to
graph transformation \cite{CorradiniHHK06}. In this approach, a rule
is a span of graphs $L \leftarrow K \to R$ and the application of
a rule to a matching of $L$ in a graph $G$ (as illustrated in the introduction)
is made of a final pullback complement followed 
by a pushout, both of them in the category of graphs. 
The sesqui-pushout
approach and ours mainly differ in the way of handling cloning.
In \cite{CorradiniHHK06}, the cloning of a node is performed 
by copying all incident edges (incoming
and outgoing edges) of the cloned node. This is a particular case of
our way of cloning nodes.  The use of polarized graphs helped us
to specify for every clone, the way incident edges can be copied.
Therefore, a sesqui-pushout rewrite step can be simulated by a
rewrite step with polarized rules, but the converse does not hold in
general.

\begin{prop}
\label{prop:rel-sqpo}
Let  $\rho$ be the SqPO rewrite rule $L \leftarrow K \to R$.
Let  $\rho'$ be the \psqpo rule $\spa{L}{l}{\polar(K)}{r}{R}$. 
Then, for every injective matching $m : L \to G$, 
\\ $\step{G}{sqpo}{H}$ implies $\step{G}{\arpsqpo}{H}$.
\end{prop}

%%%%%%%%%%%%%%%%%%%%%%%%%%%%%%%%%%%%%%%%%%%%%%%

In \cite{CorradiniHHK06}, the sesqui-pushout approach is compared 
to the classical double pushout and single pushout approaches. 
The authors show that the sesqui-pushout and the DPO approaches 
coincide under some conditions \cite[Proposition~12]{CorradiniHHK06}. 
They also show how the sesqui-pushout approach can be simulated by 
the SPO approach and they give conditions under which a SPO derivation 
can be simulated by a sesqui-pushout \cite[Proposition~13]{CorradiniHHK06}. 
So, according to Proposition~\ref{prop:rel-sqpo}, which 
shows how to simulate a sesqui-pushout step in our setting, we can
infer the same comparisons with respect to DPO and SPO for our
graph rewriting definition.

Cloning is also subject of interest in \cite{DrewesHJME06}. The
authors consider rewrite rules of the form $S \!:=\! R$ where $S$ is a
\emph{star}, i.e., $S$ is a (nonterminal) node surrounded by its adjacent
nodes together with the edges that connect them. Rewrite rules which
perform the cloning of a node are given in
\cite[Def.~6]{DrewesHJME06}. These rules show how a star can be
removed, kept identical to itself or copied (cloned) more than once. 
Here again, unlike our approach, the cloning of a node,
as in the case of the sesqui-pushout approach, copies a node
together with all its incoming and outgoing edges. 
Recently,
L\"{o}we proposed in \cite{Lowe10} a new general framework of graph
rewriting in span-categories. He has shown how classical algebraic
graph transformation approaches can be seen as instances of his
framework. Our approach, which is close to the sesqui-pushout
rewriting, could be presented too as an instance of L\"{o}we's
framework up to some particular considerations due to the use of two
kinds of graphs in our spans, namely polarized and not polarized
graphs.  Details of the instance, including the complete definitions
of abstract spans and matching of abstract spans are matter of further
investigation.

%%%%%%%%%%%%%%%%%%%%%%%%%%%%%%%%%%%%%%%%%%%%%%%

\bibliographystyle{plain}

%%%%%%%%%%%%%%%%%%%%%%%%%%%%%%%%%%%%%%%%%%%%%%%
%%%%%%%%%%%%%%%%%%%%%%%%%%%%%%%%%%%%%%%%%%%%%%%
%%%%%%%%%%%%%%%%%%%%%%%%%%%%%%%%%%%%%%%%%%%%%%%
\appendix

%%%%%%%%%%%%%%%%%%%%%%%%%%%%%%%%%%%%%%%%%%%%%%%
%%%%%%%%%%%%%%%%%%%%%%%%%%%%%%%%%%%%%%%%%%%%%%%
\section{Pushouts and final pullback complements}
\label{app}

The major result of this paper is Theorem~\ref{theo:alg-step},
which is proved as a consequence of Propositions~\ref{prop:alg-po} 
on pushouts of graphs and~\ref{prop:alg-fpbc} 
on final pullback complements of polarized graphs.
This Appendix provides detailed proofs for these 
propositions and for some apparented results. 

%%%%%%%%%%%%%%%%%%%%%%%%%%%%%%%%%%%%%%%%%%%%%%%
\subsection{Definitions} 
\label{subapp:defi} 

Pushouts (PO) and pullbacks (PB)
are basic notions of category theory \cite{MacLane},
mutually dual. They are usually depicted as follows.
  $$ PO\,:\quad \xymatrix@C=4pc{ 
  K \ar[d]_{d} \ar[r]^{r} & R \ar[d]^{h} \\ 
  D \ar[r]^{r_1} & H
  % PO
    \ar@{-}[]+L+<-6pt,+1pt>;[]+LU+<-6pt,+6pt> 
    \ar@{-}[]+U+<-1pt,+6pt>;[]+LU+<-6pt,+6pt> 
 \\ 
  } \qquad\qquad 
  PB\,:\quad \xymatrix@C=4pc{ 
  L \ar[d]_{m} & K \ar[l]_{l} \ar[d]^{d}
  % PB
    \ar@{-}[]+L+<-6pt,-1pt>;[]+LD+<-6pt,-6pt> 
    \ar@{-}[]+D+<-1pt,-6pt>;[]+LD+<-6pt,-6pt> 
   \\ 
  G & D \ar[l]_{l_1} \\ 
  } $$
In a category, the \emph{pushout} of $D\lupto{d}K\rupto{r}R$
is $D\rupto{r_1}H\lupto{h}R$ such that $h\circ r = r_1\circ d$
and for every $D\rupto{r'_1}H'\lupto{h'}R$ such that $h'\circ r = r'_1\circ d$
there is a unique $\eta:H\to H'$ such that $\eta \circ r_1 = r'_1$ and $\eta \circ h = h'$.
Dually, the \emph{pullback} of $L\rupto{m}G\lupto{l_1}D$ 
is $L\lupto{l}K\rupto{d}D$ such that $m\circ l = l_1\circ d$
and for every $L\lupto{l'}K'\rupto{d'}D$ such that $m\circ l' = l_1\circ d'$
there is a unique $\kappa:K'\to K$ such that $l \circ \kappa = l'$ and $d \circ \kappa = d'$.
Due to their universality property, 
the pushout of $D\lupto{d}K\rupto{r}R$ and the pullback of $L\rupto{m}G\lupto{l_1}D$,
when they exist, are unique up to isomorphism. 

When $L\lupto{l}K\rupto{d}D$ is the pullback of $L\rupto{m}G\lupto{l_1}D$,
as above, then $G\lupto{l_1}D\lupto{d}K $ is called a \emph{pullback complement} (PBC) 
of $G\lupto{m}L\lupto{l}K$. 
When they exist, there may be several non-isomorphic 
pullback complements of $G\lupto{l_1}D\lupto{d}K$.
The \emph{final pullback complement} (FPBC) of $G\lupto{m}L\lupto{l}K$ is 
a pullback complement $G\lupto{l_1}D\lupto{d}K $ of $G\lupto{m}L\lupto{l}K$
such that for every 
pullback complement $G\lupto{l'_1}D'\lupto{d'}K $ of $G\lupto{m}L\lupto{l}K$
there is a unique $\delta:D'\to D$ such that $\delta \circ d'=d$ 
and $l_1\circ\delta = l'_1$.
Because of its terminality property, when it exists 
the FPBC of $G\lupto{m}L\lupto{l}K$ is unique up to isomorphism.
In order to insist on this terminality property, 
in this paper the FPBC is depicted as follows.
  $$  FPBC\,:\quad \xymatrix@C=4pc{ 
  L \ar[d]_{m} & K \ar[l]_{l} \ar[d]^{d} 
    % PB
    \ar@{-}[]+L+<-6pt,-1pt>;[]+LD+<-6pt,-6pt> 
    \ar@{-}[]+D+<-1pt,-6pt>;[]+LD+<-6pt,-6pt> 
  \\ 
  G & D \ar[l]_{l_1}
  % FPBC
    \ar@{-}[]+L+<-6pt,+1pt>;[]+LU+<-6pt,+6pt> 
    \ar@{-}[]+L+<-8pt,+1pt>;[]+LU+<-8pt,+8pt> 
    \ar@{-}[]+U+<-1pt,+6pt>;[]+LU+<-6pt,+6pt> 
    \ar@{-}[]+U+<-1pt,+8pt>;[]+LU+<-8pt,+8pt> 
   \\ 
  }$$

The propositions in this Appendix are made of explicit constructions, 
which can be proved by simple verifications. 
For the constructions on graphs and polarized graphs
we use well-known results about pointwise construction of limits and colimits.
All the constructions are done up to isomorphism.

%%%%%%%%%%%%%%%%%%%%%%%%%%%%%%%%%%%%%%%%%%%%%%%
\subsection{PO, PBC and FPBC of sets}
\label{subapp:set}

In the category of sets, there are similarities between 
PO and FPBC. 
Let $\Set$ denote the category of sets,
let ``$+$'' denote the sum (or disjoint union) of sets,
and for each set $Y$ with a subset $X$ let $\ol{X}$ 
denote the complement of $X$ in $Y$, so that $Y=X+\ol{X}$.
The symbol ``$+$'' also denotes the sum of functions:
when $f_1:X_1\to Y_1$ and $f_2:X_2\to Y_2$, then 
$f_1+f_2:X_1+X_2\to Y_1+Y_2$ is defined piecewise from $f_1$ and $f_2$.  

\begin{prop}[PO of sets]
\label{app-prop:set-po} 
Let $r:K\to R$ be a function and $d:K\to D$ an inclusion, so that
$D=K+\ol{K}$. The pushout of $d$ and $r$ in $\Set$ is the following
square, where $h:R\to R+\ol{K}$ is the inclusion:
  $$ \xymatrix@C=5pc{ 
  K \ar[d]_{d} \ar[r]^{r} & R \ar[d]^{h} \\ 
  D=K+\ol{K} \ar[r]_{r_1=r+\id_{\ol{K}}} & H=R+\ol{K}
  % PO
    \ar@{-}[]+L+<-6pt,+1pt>;[]+LU+<-6pt,+6pt> 
    \ar@{-}[]+U+<-1pt,+6pt>;[]+LU+<-6pt,+6pt> 
 \\ 
  }$$
\end{prop}

\begin{proof}
Given a commutative square 
  $$ \xymatrix@C=6pc{ 
  K \ar[d]_{d} \ar[r]^{r} & R \ar[d]^{h'} \\ 
  D=K+\ol{K} \ar[r]_{r'_1} & H'
 \\ 
  }$$
let us check that there is a unique function $\eta:H\to H'$ 
such that $\eta \circ h=h'$ and $\eta \circ r_1=r'_1$.
If such a $\eta$ exists, then its restriction to $R$ is $h'$
and its restriction to $\ol{K}$ is the restriction of $r'_1$ to $\ol{K}$.
These two properties determine a function $\eta:H\to H'$
which satisfies $\eta \circ h=h'$ and $\eta \circ r_1(x)=r'_1(x)$ for each $x\in\ol{K}$.
We have to check that $\eta \circ r_1(x)=r'_1(x)$ for each $x\in K$.
This follows from the equalities 
$\eta \circ r_1\circ d= \eta \circ h\circ r = h' \circ r = r'_1\circ d$.
\end{proof}

\begin{prop}[PBC of sets]
\label{app-prop:set-pb} 
Let $l:K\to L$ be a function and $m:L\to G$ an inclusion, 
so that $G=L+\ol{L}$.
The pullback complements of $l$ and $m$ in $\Set$ are the following squares,
where $D$ is any set containing $K$ (so that $D=K+\ol{K}$), 
$d$ is the inclusion and $\ol{l}:\ol{K}\to\ol{L}$ is any function:
  $$ \xymatrix@C=6pc{ 
  L \ar[d]_{m} & K \ar[l]_{l} \ar[d]^{d}
  % PB
    \ar@{-}[]+L+<-6pt,-1pt>;[]+LD+<-6pt,-6pt> 
    \ar@{-}[]+D+<-1pt,-6pt>;[]+LD+<-6pt,-6pt> 
   \\ 
  G=L+\ol{L} & D=K+\ol{K} \ar[l]^{l_1=l+\ol{l}} \\ 
  } $$
\end{prop}

\begin{proof}
First, let us prove that the square in the proposition is a pullback square. 
Clearly it is commutative. 
Given a commutative square 
  $$ \xymatrix@C=6pc{ 
  L \ar[d]_{m} & K' \ar[l]_{l'} \ar[d]^{d'}
   \\ 
  G=L+\ol{L} & D=K+\ol{K} \ar[l]^{l_1=l+\ol{l}} \\ 
  } $$
let us check that there is a unique function $\kappa:K'\to K$
such that $l\circ \kappa=l'$ and $d\circ \kappa=d'$.
If such a $\kappa$ exists, since $d$ is the inclusion, 
the image of $d'$ is in $K$ and $\kappa(x')=d'(x')$ for each $x'\in K'$. 
This determines a function $\kappa:K'\to K$ which satisfies $d\circ \kappa=d'$.
We have to check that $l\circ \kappa=l'$, 
or equivalently, since $m$ is an inclusion, that $m\circ l\circ \kappa=m\circ l'$.
This follows from the equalities 
$m\circ l\circ \kappa = l_1\circ d\circ \kappa = l_1\circ d' = m\circ l'$.

Now, let us check that every pullback complement of $l$ and $m$ has this form.
Given a pullback square 
  $$ \xymatrix@C=6pc{ 
  L \ar[d]_{m} & K \ar[l]_{l} \ar[d]^{d'}
  % PB
    \ar@{-}[]+L+<-6pt,-1pt>;[]+LD+<-6pt,-6pt> 
    \ar@{-}[]+D+<-1pt,-6pt>;[]+LD+<-6pt,-6pt> 
   \\ 
  G=L+\ol{L} & D' \ar[l]^{l'_1} \\ 
  } $$
let us check that it has the same form as in the proposition.
Since monomorphisms are stable under pullbacks, 
the function $d'$ is a monomorphism, so that up to isomorphism 
we may assume that $d'$ is the inclusion of $K$ in $D'=K+\ol{K'}$,
where $\ol{K'}$ is the complement of $K$ in $D'$.
Since the square is commutative, the restriction of $l'_1$ to $K$ is $l$. 
Finally, let us check that $l'_1$ maps $\ol{K'}$ to $\ol{L}$:
otherwise, there is some $x'\in \ol{K'}$ such that $l'_1(x')=x$ for some $x\in L$, 
then we get a commutative square
  $$ \xymatrix@C=6pc{ 
  L \ar[d]_{m} & \{*\} \ar[l]_{*\mapsto x} \ar[d]^{*\mapsto x'}
   \\ 
  G=L+\ol{L} & D' \ar[l]^{l'_1} \\ 
  } $$
but there is no $\kappa:\{*\}\to K$ such that $\kappa(*)=x'$,
which contradicts the pullback complement property of $D'$.
\end{proof}

\begin{prop}[FPBC of sets]
\label{app-prop:set-pc} 
Let $l:K\to L$ be a function and $m:L\to G$ an inclusion, 
so that $G=L+\ol{L}$.
The FPBC of $l$ and $m$ is the following square:
 $$  \xymatrix@C=6pc{ 
  L \ar[d]_{m} & K \ar[l]_{l} \ar[d]^{d} 
    % PB
    \ar@{-}[]+L+<-6pt,-1pt>;[]+LD+<-6pt,-6pt> 
    \ar@{-}[]+D+<-1pt,-6pt>;[]+LD+<-6pt,-6pt> 
  \\ 
  G=L+\ol{L} & D=K+\ol{L} \ar[l]^{l_1=l+\id_{\ol{L}}}
  % FPBC
    \ar@{-}[]+L+<-6pt,+1pt>;[]+LU+<-6pt,+6pt> 
    \ar@{-}[]+L+<-8pt,+1pt>;[]+LU+<-8pt,+8pt> 
    \ar@{-}[]+U+<-1pt,+6pt>;[]+LU+<-6pt,+6pt> 
    \ar@{-}[]+U+<-1pt,+8pt>;[]+LU+<-8pt,+8pt> 
   \\ 
  }$$
\end{prop}

\begin{proof}
Clearly from Proposition~\ref{app-prop:set-pb}, this square is a pullback.
In order to prove that it is final, using Proposition~\ref{app-prop:set-pb},
we have to prove that for every $D'=K+\ol{K'}$
and $l'_1=l+\ol{l'}:D'\to G$ 
there is unique $\delta:D'\to D$ 
such that $\delta$ is the identity on $K$ and $l_1\circ\delta=l'_1$ on $K+\ol{K'}$.
It is easy to check that $\delta=\id_K+\ol{l'}$ is the unique function
which satisfies these properties.
  $$ \xymatrix@C=5pc{ 
  L \ar[d]_{m} & K \ar[l]_{l} \ar[d]^{d} \ar@/^3ex/[ddr]^{d'} &  \\ 
  G=L+\ol{L} & D=K+\ol{L} \ar[l]_{l_1=l+\id_{\ol{L}}} & \\ 
  && D'=K+\ol{K'} \ar@<1ex>@/^3ex/[llu]^{l'_1=l+\ol{l'}}
     \ar@<.5ex>[lu]_(.6){\;\delta=\id_K+\ol{l'}}\\ 
  }$$
\end{proof}

\begin{exmp}
\label{app-exam:set}
Here is a PBC and the FPBC of given $m$ and $l$.
Only the sets are represented,
the names of their elements describe the functions:
$m$ and $d$ are inclusions, while $l$ and $l_1$ drop the index, if any 
(every $x_i$ is mapped to $x$).
In this example, for any given $k\in\bN$  
there is a pullback complement for each $\ell\in\bN$ (in the middle) 
and the FPBC corresponds to $\ell=1$ (on the right).
$$ \begin{array}{|c|c|}
\hline 
L & K \\
\hline 
G & D\\
\hline 
\end{array}
\qquad\qquad
\begin{array}{|c|c|}
\hline 
n & n_1\dots n_{k} \\
\hline 
n & n_1\dots n_{k} \\
p & p_1\dots p_{\ell} \\
\hline 
\end{array}
\qquad\qquad
\begin{array}{|c|c|}
\hline 
n & n_1\dots n_{k} \\
\hline 
n & n_1\dots n_{k} \\
p & p \\
\hline 
\end{array} $$
\end{exmp}

%%%%%%%%%%%%%%%%%%%%%%%%%%%%%%%%%%%%%%%%%%%%%%%
\subsection{PO, PBC and FPBC of graphs} 
\label{subapp:gra} 

As in the main text, 
for each graph $X$ 
the sets of nodes and edges of $X$ are denoted 
respectively $\nod{X}$ and $\edg{X}$, 
and for all nodes $n$ and $p$ 
the set of edges from $n$ to $p$ in $X$ is denoted $\edgnp{X}{n}{p}$.
In addition for every morphism of graphs $f:X\to Y$,
for every $n_X,p_X\in \nod{X}$ and $n_Y,p_Y\in \nod{Y}$
the restrictions of $f$ are denoted: 
$$ \edgnp{f}{n_X}{p_X} : \edgnp{X}{n_X}{p_X} \to \edgnp{Y}{f(n_X)}{f(p_X)}
\; \mbox{ and }\; 
 \edgnp{f}{n_Y}{p_Y} : \sum_{n\in f^{-1}(n_Y), p\in f^{-1}(p_Y)} \edgnp{X}{n}{p} 
   \to \edgnp{Y}{n_Y}{p_Y} \;.$$

A graph $X$ may be represented informally as follows,
where $\src_X$ and $\tgt_X$ represent the source and target functions:
$$  \xymatrix@C=4pc{ 
   \edg{X} \ar@/^/[r]^{\src_X} \ar@/_/[r]_{\tgt_X}& \nod{X} \\ 
} $$
Let $\SkGr$ denote the following category (the identity arrows 
are omitted):
$$  \xymatrix@C=4pc{ 
   E \ar@/^/[r]^{\src} \ar@/_/[r]_{\tgt}& N \\ 
} $$
Then, a graph may be identified to a functor from $\SkGr$ to $\Set$.
More precisely, 
the category of graphs $\Gr$ may be identified to the category of functors from 
$\SkGr$ to $\Set$. 
It follows that limits and colimits of graphs may be computed pointwise \cite{MacLane}.

As in the main text (cf. Remark \ref{rem:gra-sum}), 
an inclusion of graphs $m:X\to Y$ 
gives rise (up to isomorphism) to a decomposition $Y=(X+\ol{X})+_e\ti{X}$. 
The edges in $\ti{X}$ are called the \emph{linking edges}.

\begin{prop}[PO of graphs. This is Proposition~\ref{prop:alg-po} in the main text]
\label{app-prop:gra-po}
Let $r:K\to R$ be a morphism and $d:K\to D$ an inclusion of graphs, 
so that $D=(K+\ol{K})+_e\ti{K}$.
The pushout of $d$ and $r$ in $\Gr$ is the following square,
where $h$ is the inclusion: 
  $$ \xymatrix@C=6pc{ 
  K \ar[d]_{d} \ar[r]^{r} & R \ar[d]^{h} \\ 
  D=(K+\ol{K})+_e\ti{K} \ar[r]_{r_1=(r+\id_{\ol{K}})+_e\ti{r}} & 
    H=(R+\ol{K})+_e\ti{R}   
  % PO
    \ar@{-}[]+L+<-6pt,+1pt>;[]+LU+<-6pt,+6pt> 
    \ar@{-}[]+U+<-1pt,+6pt>;[]+LU+<-6pt,+6pt> 
  \\ 
  }$$
where: 
\begin{itemize}
\item $ \edgnp{\ti{R}}{n}{p} = 
  \sum_{n_D\in r_1^{-1}(n),p_D\in r_1^{-1}(p)} \edgnp{\ti{K}}{n_D}{p_D}$ 
   for all $n,p\in\nod{H}$, 
\item and $\ti{r}:\ti{K}\to\ti{R}$ maps $n_D\rupto{e} p_D$ to $r_1(n_D)\rupto{e} r_1(p_D)$.
\end{itemize}
\end{prop} 

\begin{proof}
Since $D=(K+\ol{K})+_e\ti{K}$, we have 
$\nod{D}=\nod{K}+\nod{\ol{K}}$ and $\edg{D}=\edg{K}+\edg{\ol{K}}+\edg{\ti{K}}$.
Since a pushout of graphs can be computed pointwise,
let us use Proposition~\ref{app-prop:set-po}.
On nodes, we get $\nod{H}=\nod{R}+\nod{\ol{K}}$ 
with $h$ the inclusion and $r_1=r+\id_{\ol{K}}$.
On edges, we get $\edg{H}=\edg{R}+\edg{\ol{K}}+\edg{\ti{R}}$ 
where $\edg{\ti{R}}=\edg{\ti{H}}$, 
with $h$ the inclusion and $r_1=r+\id_{\edg{\ol{K}}}+\id_{\edg{\ti{K}}}$.
The source and target functions for $H$
coincide with the source and target functions for $R$ and for $\ol{K}$ 
on the subgraphs $R$ and $\ol{K}$, respectively. 
For every edge $e:n_D\to p_D$ in $\ti{K}$, its image in $H$
is $e:r_1(n_D)\to r_1(p_D)$ in $\ti{R}$.
\end{proof}

\begin{prop}[PBC of graphs]
\label{app-prop:gra-pb} 
Let $l:K\to L$ be a morphism and $m:L\to G$ an inclusion of graphs. 
The pullback complements of $l$ and $m$ in $\Gr$ are the following squares,
where $d$ is the inclusion: 
  $$ \xymatrix@C=6pc{ 
  L \ar[d]_{m} & K \ar[l]_{l} \ar[d]^{d} 
  % PB
    \ar@{-}[]+L+<-6pt,-1pt>;[]+LD+<-6pt,-6pt> 
    \ar@{-}[]+D+<-1pt,-6pt>;[]+LD+<-6pt,-6pt>  
  \\ 
  G=(L+\ol{L})+_e\ti{L} & D=(K+\ol{K})+_e\ti{K} \ar[l]^{l_1=(l+\ol{l})+_e\ti{l}} \\ 
  }$$
where: 
\begin{itemize}
\item $\ol{K}$ is any graph and 
$\ti{K}$ is any graph such that $\nodti{K}\subseteq\nod{K}+\nodol{K}$, 
\item and $\ol{l}:\ol{K}\to\ol{L}$ is any morphism of graphs and 
$\ti{l}:\ti{K}\to\ti{L}$ is any morphism of graphs which coincides with 
$l+\ol{l}$ on the nodes.
\end{itemize}
\end{prop} 

\begin{proof}
Since $G=(L+\ol{L})+_e\ti{L}$, we have 
$\nod{G}=\nod{L}+\nod{\ol{L}}$ and $\edg{G}=\edg{L}+\edg{\ol{L}}+\edg{\ti{L}}$.
Since monomorphisms are stable under pullback, 
the pullback complements of $l$ and $m$ are such that 
$d:K\to D$ is a monomorphism. 
Hence, up to isomorphism, $D=(K+\ol{K})+_e\ti{K}$,
so that $\nod{D}=\nod{K}+\nod{\ol{K}}$ and $\edg{D}=\edg{K}+\edg{\ol{K}}+\edg{\ti{K}}$, 
and $d$ is the inclusion.
We still have to prove that $l_1=(l+\ol{l})+_e\ti{l}$.
Since a pullback of graphs can be computed pointwise,
Proposition~\ref{app-prop:set-pb} tells us that 
on nodes $l_1=l+\ol{l}$ for any function $\ol{l}:\nod{\ol{K}}\to\nod{\ol{L}}$
and that on edges $l_1=l+\ol{l'}$  for any function 
$\ol{l'}:\edg{\ol{K}}+\edg{\ti{K}}\to\edg{\ol{L}}+\edg{\ti{L}}$.
In addition, an edge $e:n_D\to p_D$ in $\edg{\ol{K}}+\edg{\ti{K}}$ 
is in $\edg{\ol{K}}$ if and only if both $n_D$ and $p_D$ are in $\nod{\ol{K}}$,
and similarly an edge $e:n_G\to p_G$ in $\edg{\ol{L}}+\edg{\ti{L}}$ 
is in $\edg{\ol{L}}$ if and only if both $n_G$ and $p_G$ are in $\nod{\ol{L}}$. 
Since $l_1=l+\ol{l}$ on nodes, it follows that $\ol{l'}=\ol{l}+\ti{l}$ with 
$\ol{l}:\edg{\ol{K}}\to\edg{\ol{L}}$ and $\ti{l}:\edg{\ti{K}}\to\edg{\ti{L}}$.
\end{proof}

\begin{prop}[FPBC of graphs] 
\label{app-prop:gra-pc} 
Let $l:K\to L$ be a morphism and $m:L\to G$ an inclusion of graphs. 
The FPBC of $l$ and $m$ is the following square,
where $d$ is the inclusion: 
  $$ \xymatrix@C=6pc{ 
  L \ar[d]_{m} & K \ar[l]_{l} \ar[d]^{d} 
    % PB
    \ar@{-}[]+L+<-6pt,-1pt>;[]+LD+<-6pt,-6pt> 
    \ar@{-}[]+D+<-1pt,-6pt>;[]+LD+<-6pt,-6pt>  
\\ 
  G=(L+\ol{L})+_e\ti{L} & D=(K+\ol{L})+_e\ti{K} \ar[l]^{l_1=(l+\id_{\ol{L}})+_e\ti{l}} 
  % FPBC
    \ar@{-}[]+L+<-6pt,+1pt>;[]+LU+<-6pt,+6pt> 
    \ar@{-}[]+L+<-8pt,+1pt>;[]+LU+<-8pt,+8pt> 
    \ar@{-}[]+U+<-1pt,+6pt>;[]+LU+<-6pt,+6pt>  
    \ar@{-}[]+U+<-1pt,+8pt>;[]+LU+<-8pt,+8pt> 
  \\ 
  }$$
where: 
\begin{itemize}
\item $ \edgnp{\ti{K}}{n_D}{p_D} = 
   \edgnp{\ti{L}}{l_1(n_D)}{ l_1(p_D)}$
   for all $n_D,p_D\in\nod{D} $, 
\item and $\ti{l}$ maps $n_D\rupto{e}p_D$ to $l_1(n_D)\rupto{e}l_1(p_D)$.
\end{itemize}
\end{prop} 

\begin{proof} 
This proof generalizes the proof of Proposition~\ref{app-prop:set-pc}.
Clearly from Proposition~\ref{app-prop:gra-pb}, this square is a pullback.
In order to prove that it is final, using Proposition~\ref{app-prop:gra-pb},
we have to prove that for every $D'=(K+\ol{K'})+_e\ti{K'}$
and $l'_1=(l+\ol{l'})+_e\ti{l'}:D'\to G$ 
there is a unique $\delta:D'\to D$ 
such that $\delta$ is the identity on $K$ and $l_1\circ\delta=l'_1$.
Since $l_1=(l+\id_{\ol{L}})+_e\ti{l}$ and $l'_1=(l+\ol{l'})+_e\ti{l'}$, 
this means that $\delta=(\id_K+\ol{l'})+_e\ti{\delta}$ with 
$\ti{l}\circ\ti{\delta}=\ti{l'}:\ti{K'}\to\ti{L}$.
So, $\delta$ is uniquely determined on $K+\ol{K'}$, hence on the nodes of $D'$, 
and we still have to check that the equality $\ti{l}\circ\ti{\delta}=\ti{l'}$
determines a unique $\ti{\delta}:\ti{K'}\to\ti{K}$.
The equality $\ti{l}\circ\ti{\delta}=\ti{l'}$ is equivalent to the family of equalities
$\edgnp{\ti{l}}{n_D}{p_D} \circ \edgnp{\ti{\delta}}{n_{D'}}{p_{D'}} = 
 \edgnp{\ti{l'}}{n_{D'}}{p_{D'}}$
for all nodes $n_{D'}$ and $p_{D'}$ in $D'$,  
with $n_D=\delta(n_{D'})$ and $p_D=\delta(p_{D'})$ in $D$.
Since $\edgnp{\ti{l}}{n_D}{p_D}$ is a bijection, this is equivalent to  
$\edgnp{\ti{\delta}}{n_{D'}}{p_{D'}} = 
\edgnp{\ti{l}}{n_D}{p_D}^{-1} \circ  \edgnp{\ti{l'}}{n_{D'}}{p_{D'}}$:
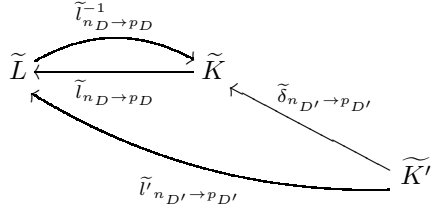
  $$ \xymatrix@C=5pc{ 
  \ti{L} \ar@/^3ex/[r]^{\edgnp{\ti{l}}{n_D}{p_D}^{-1}} & 
    \ti{K} \ar[l]^{\edgnp{\ti{l}}{n_D}{p_D}} & \\ 
  && \ti{K'} \ar@<1ex>@/^3ex/[llu]^{\edgnp{\ti{l'}}{n_{D'}}{p_{D'}}}
     \ar@<.5ex>[lu]_(.6){\;\edgnp{\ti{\delta}}{n_{D'}}{p_{D'}}} \\ 
  }$$
This determines the morphism $\ti{\delta}$, hence $\delta$.
This proof is summarized by the diagram below: 
  $$ 
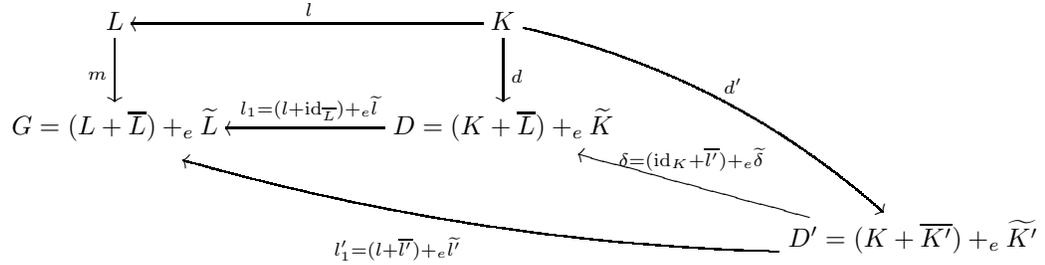
\xymatrix@C=5pc{ 
  L \ar[d]_{m} & K \ar[l]_{l} \ar[d]^{d} \ar@/^3ex/[ddr]^{d'} &  \\ 
  G=(L+\ol{L})+_e\ti{L} & D=(K+\ol{L})+_e\ti{K} \ar[l]_{l_1=(l+\id_{\ol{L}})+_e\ti{l}} & \\ 
  && D'=(K+\ol{K'})+_e\ti{K'} \ar@<1ex>@/^3ex/[llu]^{l'_1=(l+\ol{l'})+_e\ti{l'}~~~~~~~~~~~~~~~~~~~~}
     \ar@<.5ex>[lu]_(.6){\;\delta=(\id_K+\ol{l'})+_e\ti{\delta}}\\ 
  }$$
\end{proof}

\begin{exmp}
\label{app-exam:gra}
With the same conventions as in Example~\ref{app-exam:set},
here is a PBC (on the left) and the FPBC (on the right) of given $m$ and $l$ in $\Gr$.
$$  \begin{array}{|c|c|}
\hline 
 \; \xymatrix@=1pc{n & p \ar[l] \\ } \;  & 
 \; \xymatrix@=1pc{n_1 & p_1 \ar[l] & p_2  \\ } \; \\
\hline 
 \rule[20pt]{0pt}{0pt}
 \xymatrix@=1pc{n \ar@/^/[r] & p \ar[l] \ar[d] \\ & q \ar@/_/[u] \\ }  & 
 \xymatrix@=1pc{n_1 & p_1 \ar[l] \ar[d] \ar@/_/[d]\ar@/^/[d]  & p_2  \\ 
  & q_1  & q_2 \ar[u]  \\ }  \\  
\hline 
\end{array}\qquad\qquad
\begin{array}{|c|c|}
\hline 
 \; \xymatrix@=1pc{n & p \ar[l] \\ } \;  & 
 \; \xymatrix@=1pc{n_1 & p_1 \ar[l] & p_2  \\ } \; \\
\hline 
 \rule[20pt]{0pt}{0pt}
 \xymatrix@=1pc{n \ar@/^/[r] & p \ar[l] \ar[d] \\ & q \ar@/_/[u] \\ }  & 
 \xymatrix@=1pc{n_1 \ar@/^/[r] \ar@/^3ex/[rr] & p_1 \ar[l] \ar[d] & p_2 \ar[dl] \\ 
  & q \ar@/_/[u]  \ar@/_/[ur] \\ }  \\  
\hline 
\end{array}$$
\end{exmp}

%%%%%%%%%%%%%%%%%%%%%%%%%%%%%%%%%%%%%%%%%%%%%%%
\subsection{PBC and FPBC of polarized graphs}
\label{subapp:pol} 

A polarized graph $\grpol{X}=(X,X^\pm)$ may be represented informally
as the commutative diagram that follows,
where $\src_X$ and $\tgt_X$ represent the source and target functions
and the upward and downward arrows represent the inclusions, and where
$\src_{X}^+$ and $\tgt_{X}^+$ are used to enforce the fact that edges
have their source (resp. target) at a node polarized by $+$ (resp. by $-$).

$$  \xymatrix@C=.5pc@R=1pc{
    &&&&& \nod{X}^+ \ar[d] \\
   \edg{X} \ar@/^/[rrrrr]^{\src_X} \ar@/_/[rrrrr]_{\tgt_X} 
           \ar@/_/[rrrrrd]_{\tgt_{X}^-} \ar@/^/[rrrrru]^{\src_{X}^+}&&&&& \nod{X} & \\ 
    &&&&& \nod{X}^- \ar[u]  \\ 
} $$
Let $\SkGrpol$ denote the following category,
with $i^+\circ \src^+= \src $ and $i^-\circ \tgt^-= \tgt$ 
(the identity arrows are omitted):
$$  \xymatrix@C=.5pc@R=1pc{
   &&&&& N^+ \ar[d]^{i^+} \\ 
   E \ar@/^/[rrrrr]^{\src} \ar@/_/[rrrrr]_{\tgt} 
     \ar@/_/[rrrrrd]_{\tgt^-} \ar@/^/[rrrrru]^{\src^+}&&&&& N & \\ 
   &&&&& N^- \ar[u]_{i^-}  \\ 
} $$
A polarized graph may be identified to a functor from $\SkGrpol$ to $\Set$. 
More precisely, 
the category of polarized graphs $\Grpol$ may be identified 
to the category of functors from $\SkGrpol$ to $\Set$
which map $i^+$ and $i^-$ to inclusions. 
It follows that limits of polarized graphs may be 
computed pointwise.

As in the main text, we define a \emph{matching} of 
polarized graphs $m:\grpol{X} \to \grpol{Y}$ 
as a morphism of polarized graphs which 
\emph{strictly preserves} the polarization,
in the sense that 
$m(X^+)= m(X)\cap Y^+ $ and $m(X^-)= m(X)\cap Y^- $. 
Then $ \grpol{Y}= (\grpol{X}+\ol{\grpol{X}}) +_e \ti{\grpol{X}}$
(cf. Remark \ref{rem:polgra-sum}).

\begin{prop}[PBC of polarized graphs]
\label{app-prop:pol-pb} 
Let $l:\grpol{K} \to \grpol{L}$ be a morphism and 
$m:\grpol{L}\to \grpol{G}$ a matching of polarized graphs.
The pullback complements of $l$ and $m$ in $\Grpol$ are the following squares,
where $d$ is the inclusion:
  $$ \xymatrix@C=6pc{ 
  \grpol{L} \ar[d]_{m} & \grpol{K} \ar[l]_{l} \ar[d]^{d} 
  % PB
    \ar@{-}[]+L+<-6pt,-1pt>;[]+LD+<-6pt,-6pt> 
    \ar@{-}[]+D+<-1pt,-6pt>;[]+LD+<-6pt,-6pt>  
  \\
  \grpol{G}=(\grpol{L}+\ol{\grpol{L}})+_e\ti{\grpol{L}} & 
     \grpol{D}=(\grpol{K}+\ol{\grpol{K}})+_e\ti{\grpol{K}} \ar[l]^{l_1=(l+\ol{l})+_e\ti{l}} \\ 
  }$$
where:
\begin{itemize}
\item $\ol{\grpol{K}}$ is any polarized graph and 
$\ti{\grpol{K}}$ is any polarized graph such that 
$\nodti{\grpol{K}}\subseteq\nod{\grpol{K}+\ol{\grpol{K}}}$ as polarized graphs,   
\item and $\ol{l}:\ol{\grpol{K}}\to\ol{\grpol{L}}$ is any morphism of polarized graphs and  
$\ti{l}:\ti{\grpol{K}}\to\ti{\grpol{L}}$ is any morphism of polarized graphs
which coincides with $l+\ol{l}$ on nodes.
\end{itemize}
\end{prop}

\begin{proof}
This proof generalizes the proof of Proposition~\ref{app-prop:gra-pb}. 
Since monomorphisms are stable under pullback, 
the pullback complements of $l$ and $m$ are such that 
$d:\grpol{K}\to \grpol{D}$ is a monomorphism,
up to isomorphism let us assume that $d$ is the inclusion.
Let us prove that $d$ strictly preserves the polarization:
let $n_K$ be a node in $K$, with $n_L=l(n_K)$ in $L$,
such that $d(n_K)$ is polarized as $d(n_K)^+$ in $D$,
then $m(n_L)=l_1(n_K)$ is polarized as $m(n_L)^+$ in $G$, 
and since $m$ strictly preserves the polarization 
$n_L$ is also polarized as $n_L^+$ in $L$, 
which implies that $n_K$ is polarized as $n_K^+$ in $K$ because 
the square is a pullback. 
A similar result holds for nodes with negative polarization. 
So, $d$ strictly preserves the polarization, from which it follows that 
$\grpol{D}=(\grpol{K}+\ol{\grpol{K}})+_e\ti{\grpol{K}}$.
We still have to prove that $l_1=(l+\ol{l})+_e\ti{l}$.
Since a pullback of polarized graphs can be computed pointwise,
this part of the proof runs as in the proof of Proposition~\ref{app-prop:gra-pb}.
\end{proof}

\begin{prop}[FPBC of polarized graphs. 
This is Proposition~\ref{prop:alg-fpbc} in the main text]
\label{app-prop:pol-pc}
Let $l:\grpol{K} \to \grpol{L}$ be a morphism and 
$m:\grpol{L}\to \grpol{G}$ a matching of polarized graphs.
The FPBC of $l$ and $m$ is the following square, 
where $d$ is the inclusion:
  $$ \xymatrix@C=6pc{ 
  \grpol{L} \ar[d]_{m} & \grpol{K} \ar[l]_{l} \ar[d]^{d} 
    % PB
    \ar@{-}[]+L+<-6pt,-1pt>;[]+LD+<-6pt,-6pt> 
    \ar@{-}[]+D+<-1pt,-6pt>;[]+LD+<-6pt,-6pt>  
 \\
   \grpol{G}=(\grpol{L}+\ol{\grpol{L}})+_e\ti{\grpol{L}} & 
     \grpol{D}=(\grpol{K}+\ol{\grpol{L}})+_e\ti{\grpol{K}} \ar[l]^{l_1=(l+\id_{\ol{\grpol{L}}})+_e\ti{l}} 
  % FPBC
    \ar@{-}[]+L+<-6pt,+1pt>;[]+LU+<-6pt,+6pt> 
    \ar@{-}[]+L+<-8pt,+1pt>;[]+LU+<-8pt,+8pt> 
    \ar@{-}[]+U+<-1pt,+6pt>;[]+LU+<-6pt,+6pt>  
    \ar@{-}[]+U+<-1pt,+8pt>;[]+LU+<-8pt,+8pt> 
  \\
  }$$
where:
\begin{itemize}
\item $ \edgnp{\ti{K}}{n_D}{p_D} = \edgnp{\ti{L}}{l_1(n_D)}{ l_1(p_D)}$
   for all $n_D\in\nod{D}^+$ and $p_D\in\nod{D}^-$
  \\ (and otherwise $\edgnp{\ti{K}}{n_D}{p_D} = \emptyset$), 
\item and $\ti{l}$ maps $n_D\rupto{e}p_D$ to~$l_1(n_D)\rupto{e}l_1(p_D)$.
\end{itemize}
\end{prop} 

Thus, on the linking edges, the morphism $\ti{l}$ induces a bijection, 
for all $n_D\in\nod{D}^+$ and $p_D\in\nod{D}^-$:
$$  \edgnp{\ti{l}}{n_D}{p_D} :  \edgnp{\ti{K}}{n_D}{p_D} \congto 
   \edgnp{\ti{L}}{l_1(n_D)}{ l_1(p_D)} \;.
$$
If $n_D\not\in\nod{D}^+$ or $p_D\not\in\nod{D}^-$ 
then $\edgnp{\ti{K}}{n_D}{p_D}$ is empty,
but it may happen that $l_1(n_D)\in\nod{G}^+$ and $l_1(p_D)\in\nod{G}^-$
and that $\edgnp{\ti{L}}{l_1(n_D)}{ l_1(p_D)} $ is not empty.
For example:
$$ \begin{array}{|c|c|}
\hline
n^\pm & n_1^+ \qquad n_2^- \\
\hline
\xymatrix@R=1pc{n^\pm \ar[d] \\ p^\pm \\ } &
\xymatrix@R=1pc@C=1pc{n_1^+ \ar[d] & n_2^- \\ p^\pm & \\ } \\ 
\hline
\end{array}$$

\begin{proof}
This proof is similar to the proof of Proposition~\ref{app-prop:gra-pc}.
Using Proposition~\ref{app-prop:pol-pb}, 
we have to prove that for every $\grpol{D'}=(\grpol{K}+\ol{\grpol{K'}})+_e\ti{\grpol{K'}}$
and $l'_1=(l+\ol{l'})+_e\ti{l'}:\grpol{D'}\to \grpol{G}$, 
there is a unique $\delta:\grpol{D'}\to \grpol{D}$ 
such that $\delta$ is the identity on $\grpol{K}$ and $l_1\circ\delta=l'_1$.
This means that $\delta=(\id_{\grpol{K}}+\ol{l'})+_e\ti{\delta}$ with 
$\ti{l}\circ\ti{\delta}=\ti{l'}:\ti{K'}\to\ti{L}$.
So, $\delta$ is uniquely determined on $\grpol{K}+\ol{\grpol{K'}}$, 
hence on the nodes of $\grpol{D'}$, 
and we still have to check that the equality $\ti{l}\circ\ti{\delta}=\ti{l'}$
determines a unique $\ti{\delta}:\ti{K'}\to\ti{K}$.
The equality $\ti{l}\circ\ti{\delta}=\ti{l'}$ is equivalent to the family of equalities
$\edgnp{\ti{l}}{n_D}{p_D} \circ \edgnp{\ti{\delta}}{n_{D'}}{p_{D'}} = 
 \edgnp{\ti{l'}}{n_{D'}}{p_{D'}}$
for all nodes $n_{D'}$ and $p_{D'}$ in $\grpol{D'}$,  
with $n_D=\delta(n_{D'})$ and $p_D=\delta(p_{D'})$ in $\grpol{D}$.
\begin{itemize}
\item If $n_{D'}\in\nod{D'}^+$ and $p_{D'}\in\nod{D'}^-$ 
then $n_D\in\nod{D}^+$ and $p_D\in\nod{D}^-$ 
so that $\edgnp{\ti{l}}{n_D}{p_D}$ is a bijection.
Then $\edgnp{\ti{\delta}}{n_{D'}}{p_{D'}}$ is uniquely determined by 
$\edgnp{\ti{\delta}}{n_{D'}}{p_{D'}} = 
\edgnp{\ti{l}}{n_D}{p_D}^{-1} \circ \edgnp{\ti{l'}}{n_{D'}}{p_{D'}}$.
\item Otherwise $\edgnp{\ti{K'}}{n_{D'}}{p_{D'}}$ is empty, 
so that clearly $\edgnp{\ti{\delta}}{n_{D'}}{p_{D'}} $ is uniquely determined.
\end{itemize}
This yields the morphism $\ti{\delta}$, hence $\delta$.
This proof is summarized by the diagram below: 
  $$ 
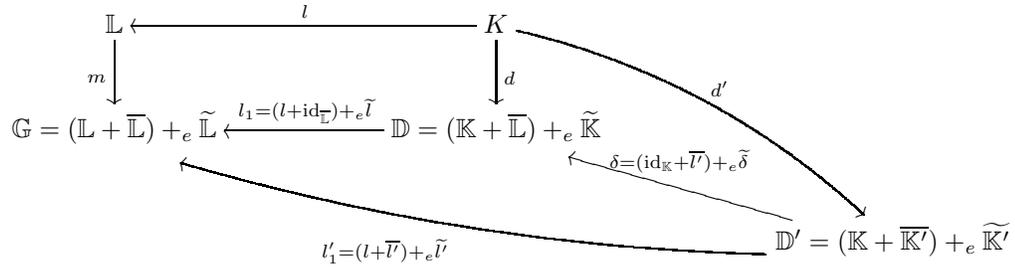
\xymatrix@C=5pc{ 
  \grpol{L} \ar[d]_{m} & K \ar[l]_{l} \ar[d]^{d} \ar@/^3ex/[ddr]^{d'} &  \\ 
  \grpol{G}=(\grpol{L}+\ol{\grpol{L}})+_e\ti{\grpol{L}} & 
    \grpol{D}=(\grpol{K}+\ol{\grpol{L}})+_e\ti{\grpol{K}} 
            \ar[l]_{l_1=(l+\id_{\ol{\grpol{L}}})+_e\ti{l}} & \\ 
  && \grpol{D'}=(\grpol{K}+\ol{\grpol{K'}})+_e\ti{\grpol{K'}}
    \ar@<1ex>@/^3ex/[llu]^{l'_1=(l+\ol{l'})+_e\ti{l'}~~~~~~~~~~~~~~~~~~~~}
     \ar@<.5ex>[lu]_(.6){\;\delta=(\id_\grpol{K}+\ol{l'})+_e\ti{\delta}}\\ 
  }$$
\end{proof}

\begin{exmp}
\label{app-exam:pol}
Here is a PBC (on the left) and the FPBC (on the right) of given $m$ and $l$ in $\Grpol$.
$$  \begin{array}{|c|c|}
\hline 
 \xymatrix@=1pc{{n}^{\pm} & {p}^\pm \ar[l] \\ }  & 
 \xymatrix@=1pc{{n_1}^\pm & {p_1}^+ \ar[l] & {p_2}^-  \\ } \\
\hline 
 \rule[20pt]{0pt}{0pt}
 \xymatrix@=1pc{{n}^\pm \ar@/^/[r] & {p}^\pm \ar[l] \ar[d] \\ 
  & {q}^\pm \ar@/_/[u] \\ }  & 
 \xymatrix@=1pc{{n_1}^\pm \ar@/^3ex/[rr] & 
  {p_1}^+ \ar[l] \ar[d] \ar[dr] & {p_2}^-  \\ 
  & {q_1}^\pm  \ar[ur] & {q_2}^\pm  \ar[u] \\ }  \\  
\hline 
\end{array}
\qquad\qquad 
\begin{array}{|c|c|}
\hline 
 \xymatrix@=1pc{{n}^{\pm} & {p}^\pm \ar[l] \\ }  & 
 \xymatrix@=1pc{{n_1}^\pm & {p_1}^+ \ar[l] & {p_2}^-  \\ } \\
\hline 
 \rule[20pt]{0pt}{0pt}
 \xymatrix@=1pc{{n}^\pm \ar@/^/[r] & {p}^\pm \ar[l] \ar[d] \\ 
  & {q}^\pm \ar@/_/[u] \\ }  & 
 \xymatrix@=1pc{{n_1}^\pm \ar@/^3ex/[rr] & {p_1}^+ \ar[l] \ar[d] & {p_2}^-  \\ 
  & {q}^\pm  \ar[ur] \\ }  \\  
\hline 
\end{array}$$
\end{exmp}

%%%%%%%%%%%%%%%%%%%%%%%%%%%%%%%%%%%%%%%%%%%%%%%
%%%%%%%%%%%%%%%%%%%%%%%%%%%%%%%%%%%%%%%%%%%%%%%
%%%%%%%%%%%%%%%%%%%%%%%%%%%%%%%%%%%%%%%%%%%%%%%
\end{document}